\documentclass[journal]{IEEEtran} 
\usepackage{amsmath}
\usepackage{amsfonts}
\usepackage{amssymb}
\usepackage{amsthm}
\usepackage{mathtools, cuted}
\usepackage{stfloats}
\usepackage[dvips]{graphicx}
\usepackage{verbatim}
\usepackage{setspace}
\usepackage{bm}
\usepackage{algorithmic} 
\usepackage[ruled,vlined]{algorithm2e}
\usepackage{cite}
\usepackage{url}

\usepackage{breakurl}
\usepackage{changepage}
\usepackage{pdfpages}
\usepackage{color}
\usepackage{blkarray, bigstrut}
\usepackage{caption}
\usepackage{subcaption}

\usepackage[labelformat=simple]{subcaption}

\newtheorem{lemma}{Lemma}
\newtheorem{proposition}{Proposition}

\newtheorem{corollary}{Corollary}
\newtheorem{remark}{Remark}

\IEEEoverridecommandlockouts

\begin{document}

\title{
    {Slow Movable Antenna System Design Based on Cell-Specific Long-Term Angular Power Spectrum}
}

\author{Ge Yan, 
        Lipeng Zhu,~\IEEEmembership{Senior Member,~IEEE,}
        Wenyan Ma,~\IEEEmembership{Member,~IEEE,}
        Rui Zhang,~\IEEEmembership{Fellow,~IEEE}

\thanks{G. Yan is with the NUS Graduate School, National University of Singapore, Singapore 119077, and also with the Department of Electrical and Computer Engineering, National University of Singapore, Singapore 117583 (e-mail: geyan@u.nus.edu). }
\thanks{L. Zhu is with the State Key Laboratory of CNS/ATM and the School of Interdisciplinary Science, Beijing Institute of Technology, Beijing 100081, China (e-mail: zhulp@bit.edu.cn). }
\thanks{W. Ma and R. Zhang are with the Department of Electrical and Computer Engineering, National University of Singapore, Singapore 117583 (e-mails: wenyan@u.nus.edu, elezhang@nus.edu.sg). }
}


\maketitle

\IEEEpeerreviewmaketitle


\begin{abstract}
    Movable antenna (MA) has recently emerged as a promising paradigm for enhancing wireless communication performance by exploiting spatial degrees of freedom through flexible antenna repositioning.
    However, most existing designs rely on short-term user-specific instantaneous/statistical channel state information (CSI), which incurs substantial channel estimation overhead and excessive complexity due to frequent antenna movement.
    To enable slow antenna movement, this paper proposes a new design framework for antenna position optimization over a much longer timescale based on cell-level statistical channel information acquired at the base station (BS).
    In particular, a \textit{cell-specific} statistical channel model is developed for MA-aided multiuser communication systems, capturing both the scattering environment and long-term user distribution. 
    Based on this model, the antenna position optimization framework for maximizing the ergodic system utility is formulated, and the covariance-eigenvalues-balancing antenna positions (CEBAP) are derived to asymptotically approximate optimal solutions by equalizing channel power across spatial eigenmodes, thereby statistically reducing users' channel correlation. 
    Notably, the CEBAP solution solely depends on the BS-side angular power spectrum (APS) of wireless channels for mobile users across the cell, which significantly alleviates the overhead of channel acquisition and antenna movement, and yet remains effective for improving various system utilities over long timescales, such as weighted sum rate and minimum signal-to-interference-plus-noise ratio (SINR). 
    To numerically solve the CEBAP, a low-complexity log-barrier penalized optimization (LOBPO) method is further proposed. 
    Simulation results based on realistic urban layouts and ray-tracing channels demonstrate consistent performance gains of the proposed MA design over fixed-position antenna (FPA) systems across different utility functions. 
    In particular, for moderately large antenna moving regions, the proposed APS-based MA positioning solution can closely approach the performance upper bound achieved by antenna position optimization based on instantaneous CSI. 

\end{abstract}



\section{INTRODUCTION}\label{sec:introduction}

    In recent years, movable antenna (MA)-aided wireless systems have attracted significant attention due to their ability to fully exploit spatial degrees of freedom (DoFs) via antenna position optimization~\cite{ref:ma-opportunities}. 
    In contrast to conventional fixed-position antennas (FPAs), an MA or MA array can dynamically adjust its position within a spatial region using various means, such as motors and micro-electromechanical systems (MEMS), thereby enabling more favorable channel conditions. 
    By leveraging this capability, considerable performance gains can be achieved for MA-aided wireless systems in various aspects~\cite{ref:zlp-ma-tutorial,ref:ma-null-steering,ref:mwy-multi-beam-ma} using the same or even a smaller number of antennas compared to conventional FPA systems~\cite{ref:ma-modeling-and-perf-analysis,ref:mwy-ma-mimo-capacity-characterization}, thus alleviating the need for a large number of costly, power-demanding radio-frequency (RF) chains. 
    Similar or related architectures have also been explored in other contexts, such as fluid antenna systems (FAS)~\cite{ref:fas-perf-mrc}, flexible-position antennas~\cite{ref:flexible-position-mimo}, rotatable antennas~\cite{ref:beixiong-wc-rotatable-antenna-oppor,ref:beixiong-tcom-rotatable-antenna-model-optm}, and pinching antennas~\cite{ref:pinching-antenna, ref:pinching-antenna-model-optm}.

    Given these promising advantages, extensive research has been conducted to reap the performance gains enabled by MAs. 
    Based on the field-response channel model tailored for MA systems, the authors in~\cite{ref:ma-modeling-and-perf-analysis, ref:zlp-ma-tutorial} demonstrated significant spatial diversity gains for narrowband single-input single-output (SISO) systems through antenna position reconfiguration, which was later extended to wideband systems in~\cite{ref:zlp-perf-ma-wideband}. 
    Meanwhile, for multiple-input multiple-output (MIMO) systems, antenna positions were optimized to reshape the channel matrix in~\cite{ref:mwy-ma-mimo-capacity-characterization}, thereby enhancing spatial multiplexing gains and maximizing system capacity. 
    In~\cite{ref:zlp-multiuser-commun-aided-by-ma,ref:ma-multiuser-joint-design,ref:ma-multiuser-discrete-global-optimal,ref:ma-wsr-max,ref:zlp-ma-near-field-statistical}, MA designs were extended to multiuser systems under both far-field and near-field channel conditions, where optimized antenna positioning mitigates inter-user channel correlation and facilitates interference suppression via beamforming. 
    Under various quality-of-service (QoS) constraints, a wide range of optimization techniques have been employed to determine antenna positions, including conventional gradient-based algorithms and successive convex approximation (SCA)~\cite{ref:mwy-ma-mimo-capacity-characterization}, as well as particle swarm optimization~\cite{ref:near-field-multiuser-ma} and deep learning methods~\cite{ref:ma-multicasting-dl}. 
    Moreover, the concept of six-dimensional movable antenna (6DMA) was introduced in~\cite{ref:6dma-modeling-and-optm-statistical,ref:6dma-discrete-optm, ref:6dma-opportunity-challenge} to further exploit antenna rotations beyond spatial translations. 
    In addition, MAs have been integrated with various technologies, such as intelligent reflecting surfaces (IRS)~\cite{ref:joint-bf-ma-irs-multi-user,ref:ma-secure-irs-isac} and cell-free MIMO~\cite{ref:ma-cell-free, ref:zyc-ma-cell-free}, and have also been applied to diverse scenarios including secure communications~\cite{ref:yaodong-ma-secure-commun-oppor} and wireless sensing~\cite{ref:wenyan-ma-sensing,ref:ma-isac-fundamentals-future,ref:mwy-reconfig-ma-for-sensing}.

    Despite these benefits, MA systems introduce new challenges due to the need for physical antenna movements. 
    In particular, most existing works optimize antenna positions based on instantaneous channel state information (CSI), which incurs substantial time overhead for channel acquisition and additional latency caused by mechanical repositioning of antennas~\cite{ref:mwy-cs-ce-ma,ref:xiaozhenyu-ma-ce-framework}. 
    Such overhead may even exceed the channel coherence time, especially in fast varying environments. 
    Besides, frequent antenna movements lead to increased energy consumption and mechanical wear, thereby hindering practical deployment. 
    To overcome these limitations, recent studies have explored antenna position optimization based on statistical CSI (S-CSI), thereby reducing repositioning overhead. 
    For example, antenna positions were optimized to maximize point-to-point MIMO channel capacity under a conventional Rician fading channel model in~\cite{ref:schober-joint-bf-ma-statistic-csi}. 
    Meanwhile, two-timescale optimization frameworks were adopted in~\cite{ref:two-timescale-ma-uplink-ula-pga} and~\cite{ref:two-timescale-ma-qqw} for maximizing the uplink and downlink sum rates of multiple users, respectively, where antenna positions are designed based on Rician fading S-CSI while beamforming is adapted to instantaneous CSI. 
    Furthermore, field-response statistical channel models were developed in~\cite{ref:my-ma-tcom,ref:jqj-6dma-stat-low-complex} for ergodic sum rate maximization over small-scale fading via antenna movement optimization, based on which the performance achieved was shown close to that based on instantaneous CSI~\cite{ref:my-ma-tcom}. 
    
    Nevertheless, the aforementioned works generally rely on \textit{user-specific} S-CSI which depends on user locations, such as the angle of departure (AoD) of the line-of-sight (LoS) channel path~\cite{ref:6dma-modeling-and-optm-statistical, ref:schober-joint-bf-ma-statistic-csi,ref:two-timescale-ma-uplink-ula-pga,ref:two-timescale-ma-qqw} and the full channel spatial covariance matrix~\cite{ref:my-ma-tcom,ref:jqj-6dma-stat-low-complex} for each user. 
    Although user locations evolve more slowly than instantaneous CSI, it is still challenging to move the antennas to track the user locations in real time. 
    Moreover, acquiring such user-specific channel information also incurs non-negligible channel estimation overhead. 
    Additionally, antenna positions optimized for different system utilities, such as sum-rate maximization, fairness enhancement, and interference mitigation, are generally distinct, which further increases the frequency of antenna repositioning required to accommodate varying system requirements. 

    To address these issues, we propose in this paper a slow antenna movement design that maximizes the ergodic performance of an MA-aided multiuser multiple-input single-output (MU-MISO) system over an extended timescale, during which mobile terminals are randomly distributed or moved within the cell. 
    In particular, we develop a new antenna position optimization framework that is robust against various system demands while requiring only the \textit{cell-specific} angular power spectrum (APS) at the base station (BS). 
    Unlike the user-specific S-CSI that changes with each user's location, the cell-specific APS is averaged over long-term wireless channels of mobile users across the cell, which is much easier to acquire and its estimation has been extensively investigated in the existing literature~\cite{ref:aps-est-hilbert,ref:aps-est-projection,ref:aps-est-time-varying}. 
    The main contributions of this paper are summarized as follows:
    \begin{itemize}
        \item Based on the statistical field-response channel model in~\cite{ref:my-ma-tcom}, a \textit{cell-specific} statistical channel model is developed for randomly distributed/moving users across the cell, which captures both the scattering environment and long-term user distribution across the cell. 
        Specifically, the cell is partitioned into smaller subregions, where the downlink channel within each subregion can be regarded as Gaussian vectors under the assumption of rich scattering around users. 
        For each subregion, the channel covariance matrix is uniquely determined by its location and dominant scatterers in the cell. 
        By weighting the subregion-specific Gaussian channel distributions with the user distribution, we obtain a cell-specific Gaussian mixture channel model, based on which the APS is defined. 
        
        \item Adopting zero-forcing (ZF) precoding at the BS, the covariance-eigenvalues-balancing antenna positions (CEBAP) are derived, which are shown to asymptotically approximate the optimal antenna positions for maximizing the ergodic system utility. 
        In particular, the CEBAP design mitigates the channel correlations among users for interference avoidance via balancing the channel power over the channel space based on the APS. 
        Although independent of any specific utility function, CEBAP can be shown effective for improving a broad class of practical performance utilities, such as weighted sum rate or minimum weighted signal-to-interference-plus-noise ratio (SINR). 
        Furthermore, a low-complexity log-barrier penalized optimization (LOBPO) method is proposed to numerically solve the CEBAP, where antenna positions are optimized via gradient ascent with non-convex constraints incorporated as penalty terms. 

        \item To validate the effectiveness of the proposed CEBAP and LOBPO methods, simulations are conducted based on realistic urban layouts in Singapore~\cite{ref:openstreetmap} and ray-tracing channel models. 
        The results demonstrate consistent performance gains over FPAs across different utilities, system configurations, and user distributions. 
        Notably, for smaller antenna moving regions, the proposed MA optimization approach yields more pronounced gains over FPAs, with the performance closely approaching that obtained via antenna position optimization catering to user-location-specific S-CSI and even instantaneous CSI. 
    \end{itemize}

    The rest of the paper is organized as follows. 
    Section~\ref{sec:system-channel-models} introduces the MA-aided system and channel models. 
    In Section~\ref{sec:prob-form-and-CEBAP}, the antenna position optimization problem is formulated and the proposed CEBAP is illustrated, while the LOBPO method is detailed in Section~\ref{sec:proposed-LOBPO}. 
    Simulation results are presented in Section~\ref{sec:perf-eval-raytrace} and conclusions are summarized in Section~\ref{sec:conclusions}.

    \textit{Notations:} 
    Boldface letters refer to vectors (lower case) or matrices (upper case). 
    For square matrix $\boldsymbol{A}$, $\text{tr}(\boldsymbol{A})$ denotes its trace and $\boldsymbol{A}^{-1}$ denotes its inverse matrix. 
    For matrix $\boldsymbol{B}$, let $\boldsymbol{B}^{T}$, $\boldsymbol{B}^{H}$, $\text{rank}(\boldsymbol{B})$, $\|\boldsymbol{B}\|_F$, and $[\boldsymbol{B}]_{nm}$ denote the transpose, conjugate transpose, rank, Frobenius norm, and the element in the $n$-th row and $m$-th column of $\boldsymbol{B}$, respectively. 
    $\boldsymbol{I}_N$ denotes the $N\times N$-dimensional identity matrix. 
    For vector $\boldsymbol{x}$, $\|\boldsymbol{x}\|_{2}$ denotes its Euclidean norm. 
    $\boldsymbol{0}_{N\times M}$ denotes the $N\times M$-dimensional zero matrix. 
    Vector $\boldsymbol{1}_{K}$ denotes the $K$-dimensional column vector with all entries equal to one. 
    For vector $\boldsymbol{x}$, $\text{Diag}(\boldsymbol{x})$ denotes the diagonal matrix whose main diagonal elements equals to $\boldsymbol{x}$. 
    For matrix $\boldsymbol{A}$, $\text{diag}(\boldsymbol{A})$ denotes the vector whose elements equals to the main diagonal elements of $\boldsymbol{A}$. 
    Sets $\mathbb{C}^{a\times b}$, $\mathbb{R}^{a\times b}$, and $\mathbb{R}_{+}^{a\times b}$ denote the space of $a\times b$-dimensional matrices with complex, real, and non-negative real elements, respectively. 
    $\mathbb{E}[\cdot]$ denotes expectation. 
    Symbol $\mathcal{CN}$ denotes the circular symmetric complex Gaussian (CSCG) distribution. 
    Symbol $\odot$ represents the Hadamard product for matrices. 
    Symbol $j = \sqrt{-1}$.


\section{System and Channel Models}\label{sec:system-channel-models}
    
    \subsection{MA-Aided MU-MISO System}\label{subsec:system-model}
        Consider a BS equipped with a two-dimensional (2D) array with ${N}$ MAs serving multiple single-FPA users within the cell, as shown in Fig.~\ref{fig:system-model}. 
        By establishing a three-dimensional (3D) local Cartesian coordinate system over the antenna plane, the position of the $n$-th MA in the array plane can be represented by a 2D vector $\boldsymbol{r}_{n} = [x_n, y_n]^T\in\mathbb{R}^{2\times 1}$, $\forall n$, where $x_n$ and $y_n$ are its coordinates along the $x$ and $y$ axies, respectively. 
        Besides, define vectors $\boldsymbol{x} = [{x}_{1}, \ldots, {x}_{N}]^T\in\mathbb{R}^{N\times 1}$ and $\boldsymbol{y} = [{y}_{1}, \ldots, {y}_{N}]^T\in\mathbb{R}^{N\times 1}$ as the $x$ and $y$ coordinates of all $N$ MAs, respectively. 
        The antenna moving region is a rectangle area on the $x$-$O$-$y$ plane centered at the origin and its sizes along the $x$ and $y$ axies are denoted as $S_{x}$ and $S_{y}$, respectively. 

        \begin{figure}[t]
            \begin{center}
                \includegraphics[scale = 0.4]{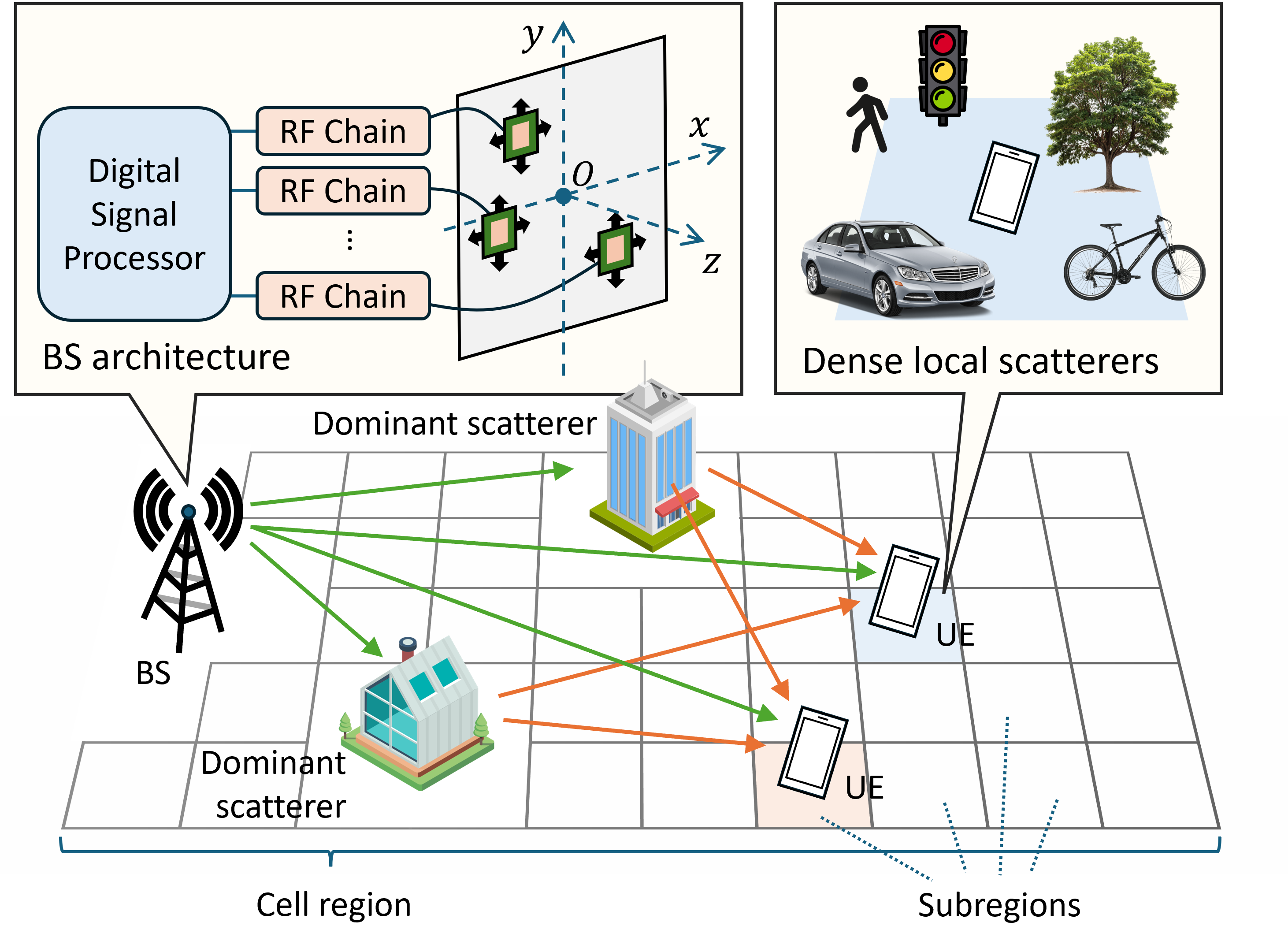}
                \vspace{-14pt}
                \caption{The MA-aided downlink MU-MISO system. }
                \label{fig:system-model}
            \end{center}
            \vspace{-4pt}
        \end{figure}

        Denote $K$ as the number of users. 
        For the $k$-th user, the precoding vector at the BS and the baseband equivalent downlink channel are denoted as $\boldsymbol{w}_{k}\in\mathbb{C}^{N\times 1}$ and $\boldsymbol{h}_{k}\in\mathbb{C}^{N\times 1}$, respectively, $1\le k\le K$. 
        By defining $\boldsymbol{W} = [\boldsymbol{w}_{1}, \ldots, \boldsymbol{w}_{K}]\in\mathbb{C}^{{N}\times K}$ and $\boldsymbol{H} = [\boldsymbol{h}_{1}, \ldots, \boldsymbol{h}_{K}]\in\mathbb{C}^{{{N}}\times{K}}$ as the downlink precoding and channel matrices, respectively, the received signal at users, denoted by $\boldsymbol{y} = [y_{1}, \ldots, y_{K}]^T\in\mathbb{C}^{K\times 1}$, is given by 
        \begin{equation}\label{def:inst-received-signal}
            \boldsymbol{y} = \boldsymbol{H}^H\boldsymbol{W}\boldsymbol{s} + \boldsymbol{z}, 
        \end{equation}
        where $\boldsymbol{s}\in\mathbb{C}^{K\times 1}$ is the transmitted signal with $\mathbb{E}[\boldsymbol{s}\boldsymbol{s}^H] = \boldsymbol{I}_{K}$, and $\boldsymbol{z}\in\mathbb{C}^{K\times 1}$ is the receiver noise following CSCG distribution $\mathcal{CN}(\boldsymbol{0}, \sigma^2\boldsymbol{I}_{K})$, while $\sigma^2$ is the average noise power.

    \subsection{Cell-Specific Gaussian-Mixture Channel Model}\label{subsec:cell-specific-stat-channel}
        
        Based on the field-response channel model tailored for the MA-aided systems~\cite{ref:my-ma-tcom}, each user's channel is composed of multiple transmit channel paths propagating from the BS towards several dominant scatterers, as shown in Fig.~\ref{fig:system-model}. 
        For users moving within a small area of several-wavelength size, the angles of departure (AoDs) of the transmit channel paths can be considered approximately constant, while the path-response coefficients are regarded as independent CSCG variables due to rich scattering around users~\cite{ref:my-ma-tcom}. 
        Thus, by dividing the cell region into $M$ non-overlapping subregions, denoted by $\mathcal{V}_{1}, \ldots, \mathcal{V}_{M}$, as shown in Fig.~\ref{fig:system-model}, the user channel within each subregion can be modeled as a CSCG vector. 
        
        Specifically, $L_{m}$ is defined as the number of transmit channel paths from the BS to the $m$-th subregion $\mathcal{V}_{m}$, while $\theta_{m, l}$ and $\phi_{m, l}$ denote the elevation and azimuth angles of the $l$-th path with respect to (w.r.t.) the MA plane, respectively, as shown in Fig.~\ref{fig:tx-path-wavevector}, with the corresponding wavevector given by $\boldsymbol{\kappa}_{m, l} = {\kappa}_{0}[\cos\theta_{m, l}\cos\phi_{m, l}, \cos\theta_{m, l}\sin\phi_{m, l}, \sin\theta_{m, l}]^{T}\in\mathbb{R}^{3\times 1}$, $1\le l\le L_{m}$, where ${\kappa}_{0} = 2\pi/\lambda$ is the carrier wavenumber with $\lambda$ being the carrier wavelength. 
        Note that we have $\boldsymbol{\kappa}_{m,l}\in\mathcal{S}_{+}\triangleq\{\boldsymbol{\kappa} = [\kappa^{x}, \kappa^{y}, \kappa^{z}]^{T} | \kappa^{z} > 0, \|\boldsymbol{\kappa}\|_{2} = \kappa_{0}\}$, $\forall m, l$. 
        Therefore, the transmit field-response vector (FRV) of the downlink channel from the $n$-th antenna to $\mathcal{V}_{m}$ is written as the following $L_{m}\times 1$ vector: 
        \begin{equation}\label{def:subregion-m-transmit-FRV}
            \boldsymbol{q}_{m}(\boldsymbol{r}_{n}) = \left[
                \exp\left(j\boldsymbol{\kappa}_{m, 1}^{T}\tilde{\boldsymbol{r}}_{n}\right), \ldots, \exp\left(j\boldsymbol{\kappa}_{m, L_{m}}^{T}\tilde{\boldsymbol{r}}_{n}\right)
            \right]^{T}, 
        \end{equation}
        where $\tilde{\boldsymbol{r}}_{n} = [\boldsymbol{r}_{n}^{T}, 0]^{T}\in\mathbb{R}^{3\times 1}$. 
        Moreover, denote the transmit path-response vector (PRV) as $\boldsymbol{\psi}_{m} = [\psi_{m, 1}, \ldots, \psi_{m, L_{m}}]^{T}\in\mathbb{C}^{L_{m}\times 1}$, which is given by $\boldsymbol{\psi}\sim\mathcal{CN}(\boldsymbol{0}, \mathrm{Diag}(\boldsymbol{\varrho}_{m}))$, where $\boldsymbol{\varrho}_{m} = [\varrho_{m, 1}, \ldots, \varrho_{m, L_{m}}]^{T}\in\mathbb{R}_{+}^{L_{m}\times 1}$ is the vector of average power gain of the transmit channel paths~\cite{ref:my-ma-tcom}. 
        Then, the channel of a random user location within the $m$-th subregion, denoted by $\boldsymbol{\hbar}_{m}$, is given by 
        \begin{equation}\label{def:subregion-m-channel-vector}
            \boldsymbol{\hbar}_{m} = \boldsymbol{Q}_{m}^{H}\boldsymbol{\psi}_{m}\sim\mathcal{CN}(\boldsymbol{0}, \boldsymbol{G}_{m}), ~\forall m, 
        \end{equation}
        where $\boldsymbol{Q}_{m} = [\boldsymbol{q}_{m}(\boldsymbol{r}_{1}), \ldots, \boldsymbol{q}_{m}(\boldsymbol{r}_{N})]\in\mathbb{C}^{L_{m}\times N}$ is the transmit field-response matrix (FRM) and $\boldsymbol{G}_{m} = \mathbb{E}_{\mathcal{V}_{m}}[\boldsymbol{\hbar}_{m}\boldsymbol{\hbar}_{m}^{H}] = \boldsymbol{Q}_{m}^{H}\mathrm{Diag}(\boldsymbol{\varrho}_{m})\boldsymbol{Q}_{m}$ is the channel covariance matrix for $\mathcal{V}_{m}$.

        To account for the user distribution, we assume indepedent users and denote the probability of a typical user in the $m$-th subregion as $\mu_{m}\ge 0, m = 1, \ldots, M$, with $\sum_{m = 1}^{M}{\mu_{m}} = 1$. 
        Given the user number $K$, the MU-MISO downlink channels can be modeled as independent and identically distributed (i.i.d.) Gaussian mixture random vectors as follows
        \begin{equation}\label{def:channel-distrbtn-gaussian-mixture}
            \boldsymbol{h}_{k}\sim\mathcal{CN}(\boldsymbol{0}, \boldsymbol{G}_{m}), ~\mathrm{w.p.} ~\mu_{m}, ~\forall k, 
        \end{equation}
        where $\text{w.p.}$ stands for ``with probability''. 
        Furthermore, the user number distribution is defined via $\zeta_{n} = \Pr(K = n)$, $n = 1, \ldots, N$, with $\sum_{n = 1}^{N}{\zeta_{n}} = 1$. 

        \begin{figure}[t]
            \begin{center}
                \includegraphics[scale = 0.4]{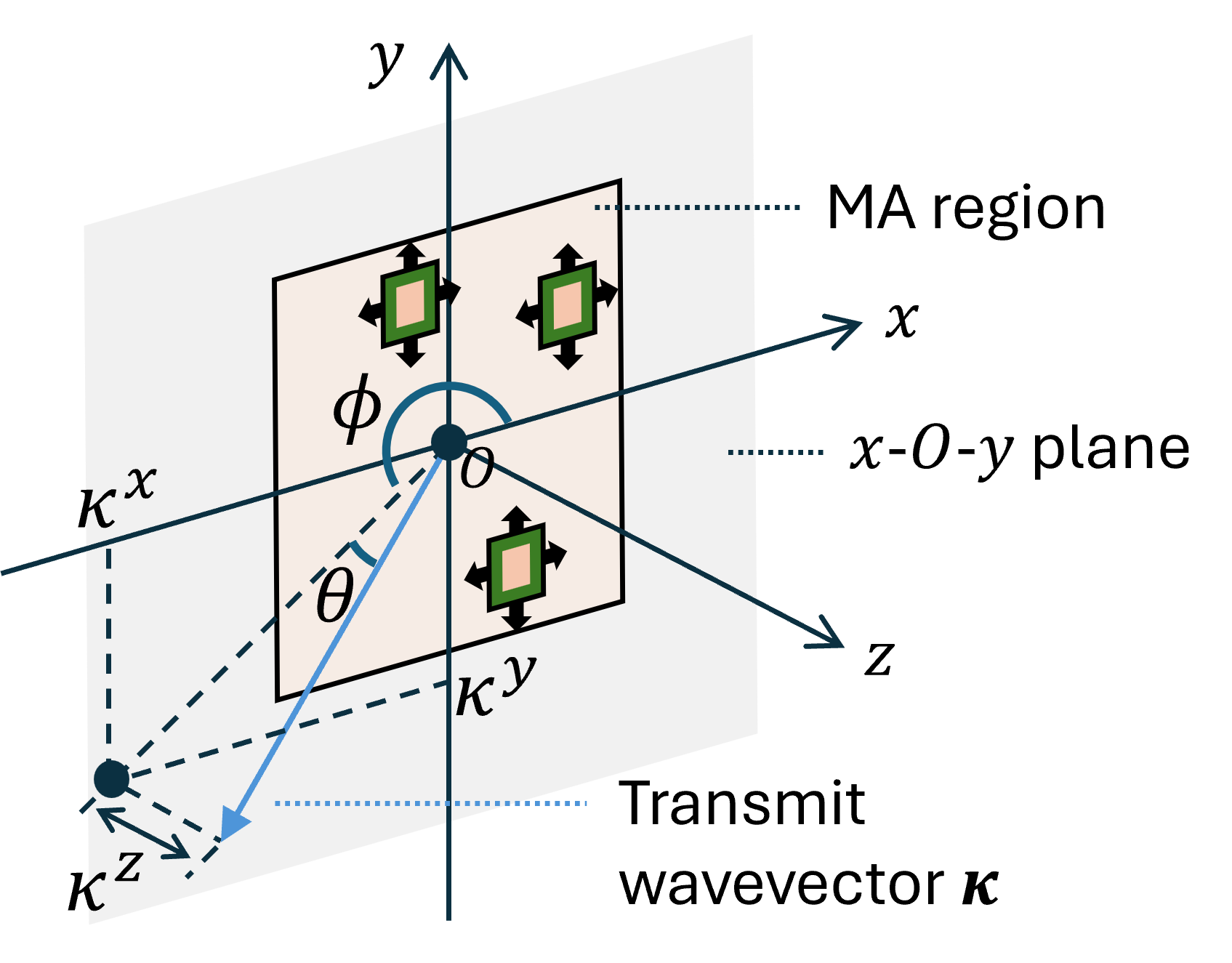}
                \vspace{-2pt}
                \caption{Transmit wavevectors. }
                \label{fig:tx-path-wavevector}
            \end{center}
            \vspace{-4pt}
        \end{figure}

    \vspace{-2pt}
    \subsection{APS and Channel Covariance Matrix}\label{subsec:aps-and-covmat}
        
        To define the APS based on the proposed Gaussian mixture channel model, the elevation and azimuth angles are equally discretized into $N_{E}$ and $N_{A}$ angles, respectively, i.e., $\vartheta_{i}$, $1\le i\le N_{E}$, and $\varphi_{i'}$, $1\le i'\le N_{A}$, which yields $L_{0} = N_{E}N_{A}$ solid angle grids over $\mathcal{S}_{+}$, denoted by $\Omega_{l}$, $l = 1, \ldots, L_{0}$. 
        With sufficiently large values of $N_{E}$ and $N_{A}$, an arbitrary wavevector $\boldsymbol{\kappa}$ falling within the $l$-th grid can be approximated by the discretized wavevector $\bar{\boldsymbol{\kappa}}_{l}\in\mathcal{S}_{+}$ defined for $\Omega_{l}$, as shown in Fig.~\ref{fig:tx-discretized-wavevector}. 
        Thus, the APS can be defined as the total power responses of transmit channel paths falling with every solid angle grid\footnote{In this paper, the elevation and azimuth angles are discretized separately for clarity and ease of interpretation, but this is not the only way. More sophisticated angular discretization schemes can be employed for uniform sampling over $\mathcal{S}_{+}$, without affecting the efficacy of the proposed approach. }. 
        Specifically, by denoting the discretized wavevectors as $\bar{\boldsymbol{\kappa}}_{l} = [\bar{\kappa}_{l}^{x}, \bar{\kappa}_{l}^{y}, \bar{\kappa}_{l}^{z}]^{T}, l = 1, \ldots, L_{0}$, the transmit power responses for subregion $\mathcal{V}_{m}$ can be equivalently represented by vector $\bar{\boldsymbol{\varrho}}_{m} = [\bar{\varrho}_{m, 1}, \ldots, \bar{\varrho}_{m, L_{0}}]^{T}\in\mathbb{R}_{+}^{L_{0}\times 1}$, where $\bar{\varrho}_{m, l}$ denotes the sum of power responses of transmit channel paths falling within the $l$-th grid for users in the $m$-th subregion, $1\le l\le L_{0}$. 
        In particular, $\bar{\varrho}_{m, l}$ is defined as
        \begin{equation}\label{def:discretized-power-response}
            \bar{\varrho}_{m, l}\triangleq\sum_{i = 1}^{L_{m}}{
                \varrho_{m, i}\cdot\mathbb{I}(\boldsymbol{\kappa}_{m, i}\in\Omega_{l})
            }, ~\forall m, l, 
        \end{equation}
        where $\mathbb{I}(\cdot)$ equals $1$ if the given condition is met and $0$ otherwise. 
        By defining $\boldsymbol{\mu} = [\mu_{1}, \ldots, \mu_{M}]^{T}\in\mathbb{R}_{+}^{M\times 1}$, the cell-specific APS can be obtained as the averaged subregion-specific power responses w.r.t. the user distribution $\boldsymbol{\mu}$, which is denoted as $\boldsymbol{b} = [b_{1}, \ldots, b_{L_{0}}]^{T}\in\mathbb{R}_{+}^{L_{0}\times 1}$ and given by
        \begin{equation}\label{def:global-aps-vec}
            \boldsymbol{b} = \boldsymbol{D}\boldsymbol{\mu}, ~\boldsymbol{D} \triangleq \left[
                \bar{\boldsymbol{\varrho}}_{1}, \ldots, \bar{\boldsymbol{\varrho}}_{M}
            \right]\in\mathbb{R}_{+}^{L_{0}\times M}, 
        \end{equation}
        with $\beta \triangleq \sum_{l = 1}^{L_{0}}b_{l}$ defined as the average channel power gain per antenna over the entire cell. 

        Moreover, by defining the FRM for discretized wavevectors as $\bar{\boldsymbol{Q}}\in\mathbb{C}^{L_{0}\times N}$, i.e., $[\bar{\boldsymbol{Q}}]_{ln} = \exp(j\bar{\boldsymbol{\kappa}}_{l}^{T}\tilde{\boldsymbol{r}}_{n})$, $\forall l, n$, it can be verified that $\boldsymbol{G}_{m} = \boldsymbol{Q}_{m}^{H}\mathrm{Diag}(\boldsymbol{\varrho}_{m})\boldsymbol{Q}_{m} = \bar{\boldsymbol{Q}}^{H}\mathrm{Diag}(\bar{\boldsymbol{\varrho}}_{m})\bar{\boldsymbol{Q}}$, $\forall m$. 
        Hence, the channel covariance matrix $\bar{\boldsymbol{G}} = \mathbb{E}\left[\boldsymbol{h}_{k}\boldsymbol{h}_{k}^{H}\right]\in\mathbb{C}^{N\times N}$ over the cell is given by 
        \begin{subequations}
            \begin{align}
                \bar{\boldsymbol{G}} & = \sum_{m = 1}^{M}{\mu_{m}\boldsymbol{G}_{m}} = \bar{\boldsymbol{Q}}^{H}\left(
                    \sum_{m = 1}^{M}{\mu_{m}\mathrm{Diag}(\bar{\boldsymbol{\varrho}}_{m})}
                \right)\bar{\boldsymbol{Q}} \\
                & = \bar{\boldsymbol{Q}}^{H}\mathrm{Diag}(\boldsymbol{b})\bar{\boldsymbol{Q}}. 
            \end{align}
        \end{subequations}
        Notably, $[\bar{\boldsymbol{G}}]_{nn} = \beta$ and $\mathrm{tr}(\bar{\boldsymbol{G}}) = N\beta$ are constant, while the off-diagonal elements of $\bar{\boldsymbol{G}}$ rely on antenna positions. 

        \begin{figure}[t]
            \begin{center}
                \includegraphics[scale = 0.4]{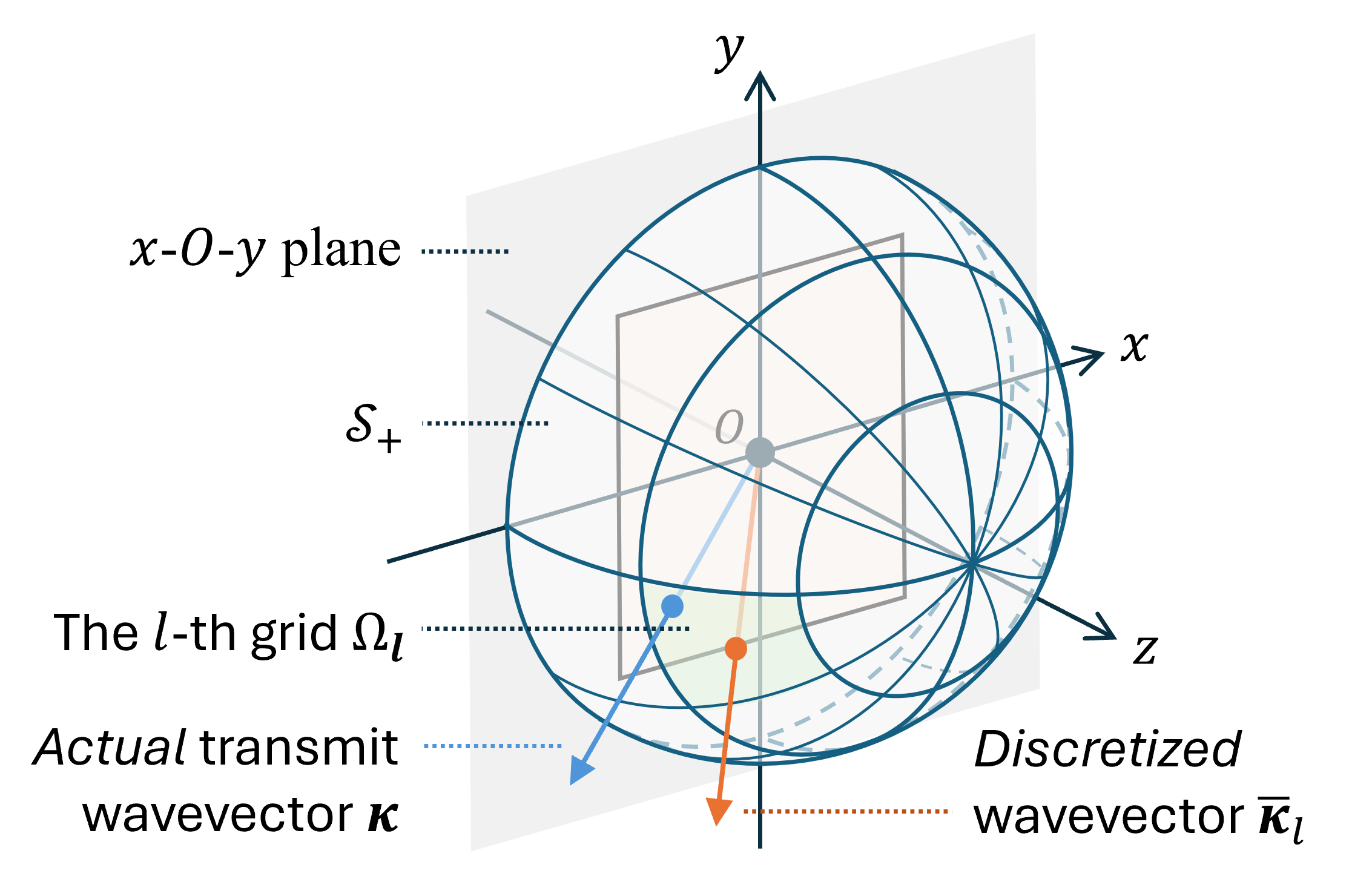}
                \caption{Discretized wavevectors. }
                \label{fig:tx-discretized-wavevector}
            \end{center}
            \vspace{-4pt}
        \end{figure}


\section{Problem Formulation and CEBAP}\label{sec:prob-form-and-CEBAP}
    In this section, the MA optimization framework for maximizing the ergodic system utility is formulated and the proposed CEBAP is illustrated. 
    
    \vspace{-4pt}
    \subsection{Two-Timescale Utility Maximization}\label{subsec:two-timescale-utility-maximization}
        Based on the Gaussian mixture channel model developed in Section~\ref{subsec:cell-specific-stat-channel}, a two-timescale optimization framework is formulated, where the precoding is designed catering to instantaneous CSI while antenna positions are optimized based on long-term cell-specific S-CSI to maximize the ergodic system utility. 
        Without loss of generality, we define the utility as a function of the number of users and their SINR, i.e., $g(\boldsymbol{\gamma}; K)$, where $\boldsymbol{\gamma} = [\gamma_{1}, \ldots, \gamma_{K}]^{T}\in\mathbb{R}_{+}^{K\times 1}$ with $\gamma_{k}$, the SINR of the $k$-th user, given by
        \begin{equation}\label{def:sinr-user-k}
            {\gamma}_{k} = \frac{\big|\boldsymbol{h}_{k}^H\boldsymbol{w}_{k}\big|^2}{\sigma^2 + \sum_{i\neq k}{\big|\boldsymbol{h}_{k}^H\boldsymbol{w}_{i}\big|^2}}, ~\forall k. 
        \end{equation}
        In particular, the utility function $g(\boldsymbol{\gamma}; K)$ is assumed non-decreasing w.r.t. each SINR $\gamma_{k}$, $\forall k$. 
        Therefore, the two-timescale optimization framework can be formulated as
        \begin{subequations}\label{prob:two-timescale-optm}
            \allowdisplaybreaks
            \begin{align}
                & \max_{\boldsymbol{x}, \boldsymbol{y}} ~\mathbb{E}_{K, \boldsymbol{H}}\left[
                    \max_{\boldsymbol{W}} ~g(\boldsymbol{\gamma}; K)
                \right] \tag{\ref{prob:two-timescale-optm}}\\
                & ~\text{s.t.} ~~\text{tr}\left(
                    \boldsymbol{W}\boldsymbol{W}^{H}
                \right)\le {P}_{T}, \label{prob-constraint:tx-power} \\
                & ~~~~~~ |x_n|\le\frac{{S}_{x}}{2}, |y_n|\le\frac{{S}_{y}}{2}, ~\forall n, \label{prob-constraint:ma-region}\\
                & ~~~~~~ \|\boldsymbol{r}_{n} - \boldsymbol{r}_{i}\|_2\ge \Delta, ~\forall n\neq i, \label{prob-constraint:ma-separation}
            \end{align}
        \end{subequations}
        where constraint~\eqref{prob-constraint:tx-power} confines the maximum transmit power ${P}_{T}$, constraint~\eqref{prob-constraint:ma-region} is resulted from the limited antenna moving region, and constraint~\eqref{prob-constraint:ma-separation} specifies the minimum inter-antenna spacing $\Delta$, which is typically set as $\Delta = \lambda/2$. 
        

    \subsection{ZF Precoding}\label{subsec:zf-precoding}
        Given $\boldsymbol{H}$, the optimal precoding depends on the utility function. 
        Nevertheless, the ZF precoding can be shown asymptotically optimal for arbitrary utility in the form of $g(\boldsymbol{\gamma}; K)$ as the transmit signal-to-noise ratio (SNR) $P_{T}/\sigma^{2}\to +\infty$, according to~\cite{ref:emil-optimal-bf}. 
        Meanwhile, it was demonstrated in~\cite{ref:my-ma-tcom} that the major gain achieved by antenna position optimization based on statistical CSI comes from its capability to reduce user channel correlation. 
        Therefore, we consider the high SNR region and adopt ZF precoding, where inter-user interference dominates the system performance, suggesting more pronounced advantages of MA-aided systems. 

        Specifically, we have $\boldsymbol{W} = \boldsymbol{H}(\boldsymbol{H}^{H}\boldsymbol{H})^{-1}\mathrm{Diag}(\boldsymbol{p})^{\frac{1}{2}}$, where $\boldsymbol{C} = \boldsymbol{H}^{H}\boldsymbol{H}\in\mathbb{C}^{K\times K}$ is invertible with probability $1$ and $\boldsymbol{p} = [p_{1}, \ldots, p_{K}]^{T}\in\mathbb{R}_{+}^{K\times 1}$ is the power allocation vector. 
        Then, the SINRs are given by $\boldsymbol{\gamma} = \boldsymbol{p}/\sigma^{2}$ and the transmit power constraint~\eqref{prob-constraint:tx-power} can be rewritten as 
        \begin{equation}\label{eq:zf-transmit-power-constraint}
            \mathrm{tr}\left(\boldsymbol{W}\boldsymbol{W}^{H}\right) = \mathrm{tr}\left(\boldsymbol{C}^{-1}\mathrm{Diag}(\boldsymbol{p})\right) = \sum_{k = 1}^{K}{\frac{p_{k}}{c_{k}}} \le {P}_{T}, 
        \end{equation}
        where $c_{k} = [\boldsymbol{C}^{-1}]_{kk}^{-1}\in\mathbb{R}$ is the reciprocal of the $k$-th diagonal element of matrix $\boldsymbol{C}^{-1}$. 
        Particularly, by using the woodbury identity~\cite{ref:woodbury-identity}, ${c}_{k}$ can be equivalently written as 
        \begin{equation}\label{eq:corr-vec-power-interpretation}
            c_{k} = \|\boldsymbol{h}_{k}\|_{2}^{2} - \boldsymbol{h}_{k}^{H}\boldsymbol{H}_{\sim k}\left(
                \boldsymbol{H}_{\sim k}^{H}\boldsymbol{H}_{\sim k}
            \right)^{-1}\boldsymbol{H}_{\sim k}^{H}\boldsymbol{h}_{k}, 
        \end{equation}
        where $\boldsymbol{H}_{\sim k} = [\boldsymbol{h}_{1}, \ldots, \boldsymbol{h}_{k - 1}, \boldsymbol{h}_{k + 1}, \ldots, \boldsymbol{h}_{K}]\in\mathbb{C}^{N\times(K - 1)}$ is the sub-matrix of $\boldsymbol{H}$ with its $k$-th column removed, $\forall k$. 
        Note that two terms in the right-hand-side of equation~\eqref{eq:corr-vec-power-interpretation} represent the total power of $\boldsymbol{h}_{k}$ and its power that falls within the subspace spanned by columns of $\boldsymbol{H}_{\sim k}$, respectively. 
        Thus, $c_{k}$ can be interpreted as the \textit{decorrelated channel power gain} for user $k$ w.r.t. other users. 
        As such, the instantaneous utility maximization problem becomes 
        \begin{equation}\label{prob:zf-power-alloc}
            \max_{\boldsymbol{p}} ~g(\boldsymbol{p}/\sigma^{2}; K), ~\mathrm{s.t.} ~\boldsymbol{p}\ge\boldsymbol{0}, ~\sum_{k = 1}^{K}{\frac{p_{k}}{c_{k}}}\le {P}_{T}. 
        \end{equation}
        Depending on the utility function, the optimal power allocation strategy varies, e.g., the water-filling algorithm for maximizing weighted sum rate. 
        By defining vector $\boldsymbol{c} = [c_{1}, \ldots, c_{K}]^{T}\in\mathbb{R}_{+}^{K\times 1}$ and denoting $\boldsymbol{p}^{\star}$ as the optimal solution to problem~\eqref{prob:zf-power-alloc}, we denote the optimal utility value as a function $\mathcal{G}(\boldsymbol{c}; K) = g(\boldsymbol{p}^{\star}/\sigma^{2}; K)$ of $\boldsymbol{c}$ and $K$. 
        The objective function for antenna position optimization in problem~\eqref{prob:two-timescale-optm} is given by 
        \begin{equation}\label{def:ma-optm-objective-w-zf}
            f(\boldsymbol{x}, \boldsymbol{y}) = \mathbb{E}_{K, \boldsymbol{H}}\left[
                \mathcal{G}(\boldsymbol{c}; K)
            \right]. 
        \end{equation}

        \begin{remark}
            Note that ZF precoding is only employed here for MA optimization, while different precoding methods can be applied in practice to cater to the specific utility function once the antennas have been deployed at the optimized positions. 
        \end{remark}

    \subsection{Asymptotic Approximation}\label{subsec:asymp-approx}
        Despite the simplification of the objective by applying ZF for problem~\eqref{prob:two-timescale-optm}, the utility-dependent power allocation and expectation over random channel realizations remain challenging for efficient implementations. 
        To address this issue, asymptotic analysis and approximations are employed for vector $\boldsymbol{c}$ in this subsection for more tractable antenna position optimization. 
        Specifically, by taking the expectation over $\boldsymbol{H}$ given $K$, $c_{k}$ can be asymptotically approximated by a constant $\rho_{K}$ detailed as follows. 
        \begin{proposition}\label{prop:asymp-approx}
            Based on the Gaussian mixture channel model, we have 
            $\mathbb{E}_{\boldsymbol{H}}[c_{k} | K] - \rho_{K} \rightarrow 0$, $\forall k$, as $N, K, L_{m}\to +\infty$ at the same rate while the total averaged channel power gain for each subregion and the ratio of maximum and minimum power responses among $\varrho_{m, 1}, \ldots, \varrho_{m, L_{m}}$ are bounded, $\forall m$. 
            In particular, $\rho_{K}$ is defined as the asymptotic decorrelated channel power gain and can be approximately solved as the unique positive solution of the following equation: 
            \begin{equation}\label{def:corr-factor-fixed-point}
                \mathrm{tr}\left[
                    \bar{\boldsymbol{G}}\left(
                        \rho_{K}\boldsymbol{I}_{N} + (K - 1)\bar{\boldsymbol{G}}
                    \right)^{-1}
                \right] = 1. 
            \end{equation}
        \end{proposition}
        \begin{proof}[Proof\textup{:}\nopunct]
            Please refer to Appendix~\ref{appendix:asymp-approx-proposition}. 
        \end{proof}
        \noindent Notably, according to Proposition~\ref{prop:asymp-approx}, the expected decorrelated channel power gain for each user becomes asymptotically identical, which is independent of the subregion-specific information but only relies on the cell-specific channel covariance matrix $\bar{\boldsymbol{G}}$. 

        Next, we move the expectation in equation~\eqref{def:ma-optm-objective-w-zf} onto $c_{k}$, $1\le k\le K$, and define function $\mathcal{G}^{\infty}(\rho; K) \triangleq \mathcal{G}(\rho\boldsymbol{1}_{K}; K)$ for $\rho\ge 0$, $\forall K$. 
        Under the conditions of Proposition~\ref{prop:asymp-approx}, we have
        \begin{subequations}
            \begin{align}
                \mathbb{E}_{\boldsymbol{H}}[\mathcal{G}(\boldsymbol{c}; K) | K] & \approx \mathcal{G}(\mathbb{E}_{\boldsymbol{H}}[\boldsymbol{c} | K]; K) \approx \mathcal{G}(\rho_{K}\boldsymbol{1}_{K}; K) \\
                & = \mathcal{G}^{\infty}(\rho_{K}; K). 
            \end{align}
        \end{subequations}
        Hence, we propose to approximate $f(\boldsymbol{x}, \boldsymbol{y})$ as
        \begin{equation}
            f(\boldsymbol{x}, \boldsymbol{y}) = \mathbb{E}_{K}\left[
                \mathbb{E}_{\boldsymbol{H}}\left[
                    \mathcal{G}(\boldsymbol{c}; K) | K
                \right]
            \right]\approx\mathbb{E}_{K}\left[
                \mathcal{G}^{\infty}(\rho_{K}; K)
            \right], 
        \end{equation}
        which further simplifies the objective function for antenna position optimization by approximating the expectation over random channel realizations. 
        
        Nevertheless, the asymptotic approximation requires knowledge of $\rho_{K}$, $\forall K$, which is generally difficult to obtain in practice. 
        For $K = 1$, we have $\rho_{1} = \mathrm{tr}(\bar{\boldsymbol{G}}) = (N\beta)$, representing the expected channel power gain with zero interference. 
        For $K > 1$, $\rho_{K}$ is implicitly defined via equation~\eqref{def:corr-factor-fixed-point} and cannot be solved in closed form. 
        To address this issue, the eigenvalue decomposition is employed to the Hermitian channel covariance matrix as $\bar{\boldsymbol{G}} = \boldsymbol{U}\mathrm{Diag}(\boldsymbol{\lambda})\boldsymbol{U}^{H}$, where $\boldsymbol{U}\in\mathbb{C}^{N\times N}$ is unitary and $\boldsymbol{\lambda} = [\lambda_{1}, \ldots, \lambda_{N}]^{T}\in\mathbb{R}_{+}^{N\times 1}$ denotes the vector of eigenvalues. 
        Then, the left-hand-side of equation~\eqref{def:corr-factor-fixed-point}, denoted by function $\xi_{K}(\rho_{K})$, can be equivalently written as follows:
        \begin{equation}\label{def:corr-facotr-xi-func}
            \xi_{K}(\rho) = \sum_{n = 1}^{N}\frac{\lambda_{n}}{\rho + (K - 1)\lambda_{n}}, ~\rho\ge 0, ~K = 2, \ldots, N. 
        \end{equation}
        Therefore, $\rho_{K}$ is the unique positive solution of equation $\xi_{K}(\rho) = 1$, which can be efficiently solved via the Newton's method~\cite{ref:newton-method}. 
        Specifically, define $J_{K}(\rho) = \partial{\xi_{K}(\rho)}/\partial{\rho}$ as the derivative of $\xi_{K}(\rho)$ w.r.t. $\rho \ge 0$, which is given by
        \begin{equation}\label{def:corr-factor-xi-func-deriv}
            J_{K}(\rho) = - \sum_{n = 1}^{N}{
                \frac{\lambda_{n}}{(\rho + (K - 1)\lambda_{n})^{2}}
            }.
        \end{equation}
        By denoting $\rho_{K}^{(0)} = 0$ as initialization and $I_{\text{c}}$ as the maximum number of iterations, $\rho_{K}$ can be solved via the following iterative updates: 
        \begin{equation}\label{def:newtons-iterations}
            \rho_{K}^{(i + 1)} = \rho_{K}^{(i)} - J_{K}(\rho_{K}^{(i)})^{-1}\left[\xi_{K}(\rho_{K}^{(i)}) - 1\right], 
        \end{equation}
        where $\rho_{K}^{(i)}$ is the value updated in the $i$-th iteration, $0 \le i\le I_{\text{c}} - 1$. 
        The convergence of the Newton's method can be easily verified as $\xi_{K}(\rho)$ is convex and decreasing with $\rho\ge 0$, which is detailed in Appendix~\ref{appendix:newton-convergence}. 

    \subsection{CEBAP Solution}\label{subsec:cebap}
        From equation~\eqref{def:corr-factor-fixed-point} and definition of function $\mathcal{G}^{\infty}$, it can be easily verified that $\rho_{K}$ decreases with the user number $K$ while $\mathcal{G}^{\infty}(\rho; K)$ strictly increases with $\rho$. 
        By considering the worst case where the number of users equals that of BS antennas, i.e., replacing $\rho_{K}$ with $\rho_{N}$, we have $\mathcal{G}^{\infty}(\rho_{K}; K)\ge\mathcal{G}^{\infty}(\rho_{N}; K)$ and thus 
        \begin{equation}\label{eq:worse-case-lowerbound}
            f(\boldsymbol{x}, \boldsymbol{y}) \approx \mathbb{E}_{K}\left[
                \mathcal{G}^{\infty}(\rho_{K}; K)
            \right] \ge \mathbb{E}_{K}\left[
                \mathcal{G}^{\infty}(\rho_{N}; K)
            \right]. 
        \end{equation}
        Note that the right-hand-side of the inequality in~\eqref{eq:worse-case-lowerbound} serves as an approximate worst-case lower-bound on the objective function $f(\boldsymbol{x}, \boldsymbol{y})$ and is strictly decreasing with $\rho_{N}$. 
        Therefore, instead of optimizing $f(\boldsymbol{x}, \boldsymbol{y})$ that explicitly relies on the utility function and user number distribution, we propose to maximize the asymptotic decorrelated channel power gain $\rho_{N}$, which equivalently maximizes the lower-bound and therefore intermediately increases $f(\boldsymbol{x}, \boldsymbol{y})$. 
        Specifically, the antenna position optimization problem for ergodic utility maximization is finally relaxed to the following problem:
        \begin{equation}\label{prob:cebap}
            \max_{\boldsymbol{x}, \boldsymbol{y}} ~\rho_{N}, ~\mathrm{s.t.} ~\eqref{prob-constraint:ma-region},\eqref{prob-constraint:ma-separation}. 
        \end{equation}
        Since $\rho_{N}$ is determined by $N$ and $\bar{\boldsymbol{G}}$, problem~\eqref{prob:cebap} can be solved with the knowledge of APS alone at the BS, which is necessary for computing $\bar{\boldsymbol{G}}$, regardless of the transmit AoDs and power responses for all subregions, the user distribution, or even the utility function. 
        This significantly reduces the channel acquisition overhead and enhances robustness of the system performance against various demands. 

        The optimal solution for problem~\eqref{prob:cebap}, denoted by $(\boldsymbol{x}^{\star}, \boldsymbol{y}^{\star})$, is referred to as the CEBAP and will be numerically solved later by the proposed LOBPO method in Section~\ref{sec:proposed-LOBPO}. 
        In particular, it can be shown that CEBAP effectively balances the eigenvalues of the channel covariance matrix. 
        Specifically, by noting that $\rho_{N}$ depends on the eigenvalues $\lambda_{n}$, $1\le n\le N$, and that $\sum_{n = 1}^{N}\lambda_{n} = \mathrm{tr}(\bar{\boldsymbol{G}}) = N\beta$, the following lemma and corollary can be easily verified and thus their proofs are omitted. 
        \begin{lemma}\label{lemma:corr-factor-w-eigens}
            For arbitrary two indices $n\neq i$ satisfying $\lambda_{n} \ge \lambda_{i}$, $\rho_{N}$ decreases with $\delta = \lambda_{n} - \lambda_{i}\ge 0$ provided that all other eigenvalues are fixed. 
        \end{lemma}
        \begin{corollary}\label{coro:corr-factor-limits}
            The upper bound on $\rho_{N}$ is obtained when $\lambda_{n} = \beta$, $\forall n$, are identical, which yields $\rho_{N}^{\mathrm{max}} = \beta$. 
            In contrast, we have $\rho_{N}\to 0$ as $\lambda_{1}\to N\beta$ and $\lambda_{n}\to 0$ for $n\ge 2$. 
        \end{corollary}
        Therefore, $\rho_{N}$ is generally smaller when the eigenvalues of the channel covariance matrix are more diverse, while maximizing $\rho_{N}$ enforces more balanced eigenvalues even though the upper bound $\rho_{N}^{\text{max}}$ is not necessarily available. 
        As such, the total spectrum power is spread out over the $N$-dimensional channel space instead of being concentrated within a smaller subspace, thereby reducing the correlation between users' channels and increasing their decorrelated channel power gain to the most extent, which explains why CEBAP is consistently effective for improving various utilities. 
        Furthermore, it can be verified that $\rho_{N}$ is independent of the transmit power $P_{T}$, indicating that the same CEBAP can be applied regardless of the system operating SNR.

    \subsection{Discussions}\label{subsec:discussions}
        In this subsection, we consider a von Mises-Fisher (vMF)-type APS~\cite{ref:bayesian-vmf} to evaluate the effectiveness of CEBAP. 
        This choice provides both practical insights and analytical generality, because vMF distributions are widely used to model clustered users or scatterer directions~\cite{ref:spatial-corr-vmf, ref:general-spatial-corr-3d-vmf}, and general spherical functions can be accurately approximated by mixtures of vMF components~\cite{ref:clustering-vmf}, thus making the subsequent analysis applicable to arbitrary APS. 
        
        Specifically, denote $\omega_{l}$ as the surface area of the $l$-th solid angle grid $\Omega_{l}$ over $\mathcal{S}_{+}$, 
        $\forall l$. 
        Then, we define the angular power spectrum density (APSD)\footnote{Since the surface areas of discretized solid angle grids are not uniform, spectrum power for grids with larger surface areas tend to be higher, indicating that the APS depends on the angular discretization scheme. In contrast, the APSD is invariant given the scattering environment and user distribution with sufficiently large $N_{E}$ and $N_{A}$. Thus, the angular distribution of channel power gain across the cell can be reflected more properly by evaluating the APSD, despite that it is is easier to formulate and solve the proposed CEBAP based on the APS. } at grid $\Omega_{l}$ by dividing $b_{l}$ with the surface area $\omega_{l}$, representing the channel power density along $\bar{\boldsymbol{\kappa}}_{l}$, which is assumed to follow the vMF distribution:
        \begin{equation}\label{def:vmf-type-aps}
            \frac{b_{l}}{\omega_{l}} = \frac{1}{B}\exp(\boldsymbol{\nu}^{T}\bar{\boldsymbol{\kappa}}_{l}), ~l = 1, \ldots, L_{0}, 
        \end{equation}
        where $\boldsymbol{\nu} = [\nu_x, \nu_{y}, \nu_{z}]^{T}\in\mathbb{R}^{3\times 1}$, while $B$ is a normalization factor such that $\sum_{l = 1}^{L_{0}}{b_{l}} = \beta$. 
        As such, $\boldsymbol{b}$ can be regarded as the APS obtained for a user cluster in the direction of $\boldsymbol{\nu}$, where $\nu_{0} = \|\boldsymbol{\nu}\|_{2}$ is defined as the concentration factor and $\hat{\boldsymbol{\nu}} = \boldsymbol{\nu}/\nu_{0}$ denotes the cluster direction. 
        In Fig.~\ref{subfig:vmf-apsd-wide} and~\ref{subfig:vmf-apsd-narrow}, two APSDs are shown as examples given $\hat{\boldsymbol{\nu}} = [0, 1/2, \sqrt{3}/2]^{T}$ for $\nu_{0} = 0.1$ and $\nu_{0} = 1$, respectively, which is projected from the hemisphere $\mathcal{S}_{+}$ onto the $x$-$O$-$y$ plane while $\kappa^{x}$ and $\kappa^{y}$ are normalized by $\kappa_{0}$. 
        The total power of the APSDs is normalized to $1$ and the carrier frequency is set as $f_{c} = 5$ GHz, yielding $\lambda = 6$ cm and $\kappa_{0} = 104.72$. 
        Notably, the power density is more focused around $\hat{\boldsymbol{\nu}}$ for a larger concentration factor $\nu_{0}$, while the power density along directions far from $\hat{\boldsymbol{\nu}}$ becomes negligible, indicating a more concentrated user cluster and thus higher user channel correlation. 

        \begin{figure}[t!]
            \centering
            {
                \begin{subfigure}[t]{0.235\textwidth}
                    \centering
                    \includegraphics[scale = 0.33]{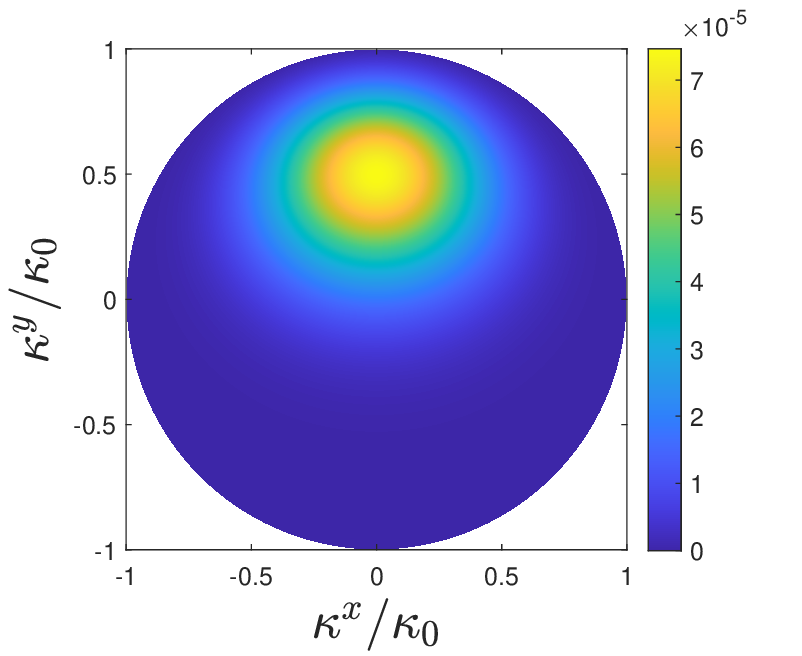}
                    \caption{$\nu_{0} = 0.1$. }
                    \label{subfig:vmf-apsd-wide}
                \end{subfigure}
                \begin{subfigure}[t]{0.24\textwidth}
                    \centering
                    \includegraphics[scale = 0.33]{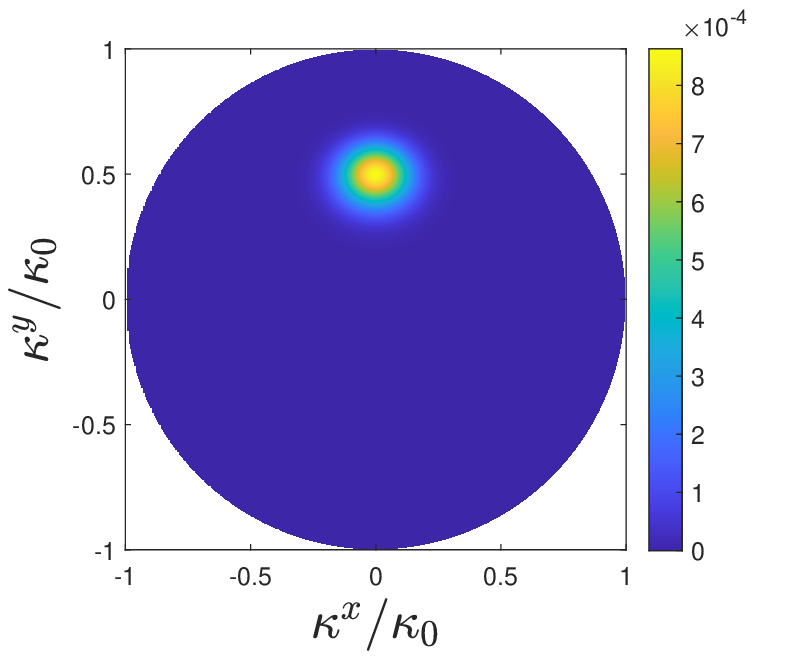}
                    \caption{$\nu_{0} = 1$. }
                    \label{subfig:vmf-apsd-narrow}
                \end{subfigure}
            }
            \vspace{-2pt}
            \caption{Examples for vMF-type APSDs given $\kappa_{0} = 104.72$ and $\hat{\boldsymbol{\nu}} = [0, 1/2, \sqrt{3}/2]^{T}$ with normalized total power. }
            \label{fig:vmf-apsd}
            \vspace{-2pt}
        \end{figure}

        Next, assuming $\nu_{z} > 0$ and sufficiently large $N_{E}$, $N_{A}$, and $\nu_{0}$ for~\eqref{def:vmf-type-aps} and defining $\boldsymbol{\delta}_{n, i} = \tilde{\boldsymbol{r}}_{n} - \tilde{\boldsymbol{r}}_{i}$ as the relative antenna position vector, $\forall n, i$, it can be shown that
        \begin{equation}\label{eq:vmf-corr-sinc}
            \left[\bar{\boldsymbol{G}}\right]_{ni} \approx \frac{4\pi\kappa_{0}^{2}}{B}\mathrm{sinc}(\kappa_{0}d_{n,i}), ~\forall n, i, 
        \end{equation}
        where $\mathrm{sinc}(0) = 1$ and $\mathrm{sinc}(x) = \sin{x}/x$ for any non-zero $x\in\mathbb{C}$, and $d_{n,i}$ is given by 
        \begin{equation}\label{def:vmf-corr-complex-dist}
            d_{n,i} = \sqrt{
                (\boldsymbol{\delta}_{n,i} + j\boldsymbol{\nu})^{T}(\boldsymbol{\delta}_{n,i} + j\boldsymbol{\nu})
            }\in\mathbb{C}. 
        \end{equation}
        Derivations for~\eqref{eq:vmf-corr-sinc} and~\eqref{def:vmf-corr-complex-dist} are given in Appendix~\ref{appendix-subsec:vmf-corr-derivation}. 
        In particular, it can be verified that as $\|\boldsymbol{\delta}_{n, i}\|_{2}\rightarrow +\infty$ given vector $\boldsymbol{\nu}$, we have
        \begin{equation}\label{eq:vmf-corr-sparse-limit}
            4\pi\kappa_{0}^{2}B^{-1}\mathrm{sinc}(\kappa_{0}d_{n, i})\rightarrow 0, 
        \end{equation}
        which is shown in Appendix~\ref{appendix-subsec:vmf-corr-sparse-limit-proof}. 
        Thus, all non-diagonal elements of matrix $\bar{\boldsymbol{G}}$ vanish if all antennas are separated with sufficiently large distances, leading to $\bar{\boldsymbol{G}}\approx\beta\boldsymbol{I}_{N}$ and $\rho_{N}\approx\rho_{N}^{\text{max}} = \beta$. 
        This reveals that any sparse array layout overspreading a sufficiently large MA region achieves approximately the same performance as the proposed CEBAP. 
        On the other hand, as $\nu_{0}\rightarrow +\infty$ given a finite MA region, i.e., the spectrum power is mostly concentrated along $\hat{\boldsymbol{\nu}}$, we have
        \begin{subequations}\label{eq:vmf-corr-concentrated-limit}
            \begin{align}
                & 4\pi\kappa_{0}^{2}B^{-1}\mathrm{sinc}(\kappa_{0}d_{n, i}) \rightarrow \beta\exp(-j\kappa_{0}\boldsymbol{\delta}_{n, i}^{T}\hat{\boldsymbol{\nu}}) \\
                & ~~~~ ~~~~ ~~~~ = \beta\exp(-j\kappa_{0}\tilde{\boldsymbol{r}}_{n}^{T}\hat{\boldsymbol{\nu}})\exp(j\kappa_{0}\tilde{\boldsymbol{r}}_{i}^{T}\hat{\boldsymbol{\nu}}), ~\forall n, i, 
            \end{align}
        \end{subequations}
        whose proof is given in Appendix~\ref{appendix-subsec:vmf-corr-concentrated-limit-proof}. 
        Under such conditions, $\bar{\boldsymbol{G}}\approx \beta\boldsymbol{v}\boldsymbol{v}^{H}$ is approximately a rank-one matrix, where $\boldsymbol{v} = [\exp(j\kappa_{0}\tilde{\boldsymbol{r}}_{1}^{T}\hat{\boldsymbol{\nu}}), \ldots, \exp(j\kappa_{0}\tilde{\boldsymbol{r}}_{N}^{T}\hat{\boldsymbol{\nu}})]^{H}\in\mathbb{C}^{N\times 1}$ is the steering vector along direction $\hat{\boldsymbol{\nu}}$, yielding $\rho_{N}\approx 0$ regardless of antenna positions. 
        Then, the system performance is independent of antenna positions because all users' channel vectors are always parallel. 
        Based on the analysis above, it could be expected that the advanatges of CEBAP against FPAs become pronounced when the prescribed MA region is moderately large and the APS is moderately concentrated, which typically yields a high-rank covariance matrix with unbalanced eigenvalues, as will be verified via simulations in Section~\ref{sec:perf-eval-raytrace}.


\section{Proposed LOBPO Method for Solving CEBAP}\label{sec:proposed-LOBPO}
    In this section, the proposed LOBPO method is illustrated for numerically solving the CEBAP from problem~\eqref{prob:cebap}. 
    Particularly, the non-convex constraints~\eqref{prob-constraint:ma-region} and~\eqref{prob-constraint:ma-separation} are incorporated into log-barrier penalty functions and the gradient ascent algorithm is then applied to maximize the penalized objective.

    \subsection{Log-Barrier Penalty}\label{subsec:log-barrier-penalty}
        Define $\mathcal{S}_{\text{MA}}$ as the feasible set for the antenna position duplet $(\boldsymbol{x}, \boldsymbol{y})$ that satisfies constraints~\eqref{prob-constraint:ma-region} and~\eqref{prob-constraint:ma-separation}. 
        Based on set $\mathcal{S}_{\text{MA}}$, the log-barrier function $\mathcal{L}(\boldsymbol{x}, \boldsymbol{y})$ is defined as follows:
        \begin{equation}\label{def:log-barrier-func}
            \mathcal{L}(\boldsymbol{x}, \boldsymbol{y}) = \left\{
                \begin{array}{ll}
                    \mathcal{L}_{\mathcal{S}}(\boldsymbol{x}, \boldsymbol{y}), & (\boldsymbol{x}, \boldsymbol{y})\in\text{int}(\mathcal{S}_{\text{MA}}), \\
                    -\infty, & \text{otherwise}, 
                \end{array}
            \right.
        \end{equation}
        where $\text{int}(\mathcal{S}_{\text{MA}})$ denotes the interior of $\mathcal{S}_{\text{MA}}$, while function $\mathcal{L}_{\mathcal{S}}(\boldsymbol{x}, \boldsymbol{y})$ is given by 
        \begin{equation}\label{def:log-barrier-feasible-func}
            \begin{aligned}
                & \mathcal{L}_{\mathcal{S}}(\boldsymbol{x}, \boldsymbol{y}) \triangleq \frac{1}{N^{2}}\sum_{
                    1 \le n < i \le {N}
                }{\ln{\left(
                    \|\boldsymbol{r}_{n} - \boldsymbol{r}_{i}\|_2^2 - \Delta^2
                \right)}} \\
                & ~~~~~~ + \frac{1}{N}\left[
                    \sum_{n = 1}^{{N}}{\ln{\left(
                        \frac{{S}_{x}^2}{4} - x_n^2
                    \right)} + \ln{\left(
                        \frac{{S}_{y}^2}{4} - y_n^2
                    \right)}}
                \right], 
            \end{aligned}
        \end{equation}
        which approaches $-\infty$ as $(\boldsymbol{x}, \boldsymbol{y})$ is close to the boundary of $\mathcal{S}_{\text{MA}}$. 
        By applying $\mathcal{L}(\boldsymbol{x}, \boldsymbol{y})$ as a penalty to replace constraints~\eqref{prob-constraint:ma-region} and~\eqref{prob-constraint:ma-separation}, problem~\eqref{prob:cebap} can be approximated as 
        \begin{equation}\label{prob:cebap-log-barrier-relaxed}
            \max_{\boldsymbol{x}, \boldsymbol{y}} ~F(\boldsymbol{x}, \boldsymbol{y}) \triangleq \rho_{N} + \alpha\mathcal{L}(\boldsymbol{x}, \boldsymbol{y}), 
        \end{equation}
        where $\alpha > 0$ is the penalty weight. 
        Note that if $(\boldsymbol{x}, \boldsymbol{y})$ does not satisfy constraints~\eqref{prob-constraint:ma-region} and~\eqref{prob-constraint:ma-separation}, we have $F(\boldsymbol{x}, \boldsymbol{y}) = \mathcal{L}(\boldsymbol{x}, \boldsymbol{y}) = -\infty$, indicating that $(\boldsymbol{x}, \boldsymbol{y})$ is not the optimal solution for problem~\eqref{prob:cebap-log-barrier-relaxed}. 
        Therefore, the feasibility of the optimal solution for problem~\eqref{prob:cebap-log-barrier-relaxed} to constraints~\eqref{prob-constraint:ma-region} and~\eqref{prob-constraint:ma-separation} is guaranteed. 
        Besides, the penalty term $\alpha\mathcal{L}_{\mathcal{S}}(\boldsymbol{x}, \boldsymbol{y})$ vanishes for any $(\boldsymbol{x}, \boldsymbol{y})\in\text{int}(\mathcal{S}_{\text{MA}})$ as $\alpha\to 0^{+}$. 
        Thus, the optimal solution for problem~\eqref{prob:cebap-log-barrier-relaxed} approximately maximizes $\rho_{N}$ given a sufficiently small $\alpha$ and suboptimally solves problem~\eqref{prob:cebap}. 

    \subsection{Gradient Ascent for Solving Problem~\eqref{prob:cebap-log-barrier-relaxed}}\label{subsec:gradient-ascent}
        Given the unconstrained log-barrier penalized problem~\eqref{prob:cebap-log-barrier-relaxed}, the gradient ascent algorithm can be applied to solve the CEBAP. 
        By defining matrix $\boldsymbol{Y}_{N} = \boldsymbol{I}_{N} + \rho_{N}^{-1}(N - 1)\bar{\boldsymbol{G}}\in\mathbb{C}^{N\times N}$ and leveraging the equation $\xi_{N}(\rho_{N}) = 1$, the gradients of $\rho_{N}$ w.r.t. $\boldsymbol{t}\in\{\boldsymbol{x}, \boldsymbol{y}\}$ can be obtained as
        \begin{subequations}\label{def:corr-factor-grads-wrt-positions}
            \begin{align}
                \nabla_{\boldsymbol{t}}{\rho_{N}} & \triangleq \frac{\mathrm{d}{\rho_{N}}}{\mathrm{d}\boldsymbol{t}} = -\left[\left(
                    \frac{\partial\xi_{N}(\rho)}{\partial\rho}
                \right)^{-1}\frac{\partial\xi_{N}(\rho)}{\partial\boldsymbol{t}}\right]
                \Bigg|_{\rho = \rho_{N}} \\
                & = \frac{
                    2\rho_{N}\mathrm{Re}\left[
                        \mathrm{diag}\left(\boldsymbol{Y}_{N}^{-2}\boldsymbol{S}^{t}\right)
                    \right]
                }{
                    \mathrm{tr}\left(\boldsymbol{Y}_{N}^{-2}\bar{\boldsymbol{G}}\right)
                }\in\mathbb{R}^{N\times 1}, 
            \end{align}
        \end{subequations}
        where $\boldsymbol{S}^{t}\in\mathbb{C}^{N\times N}$ is given by 
        \begin{equation}\label{def:corr-factor-grads-Smat}
            \boldsymbol{S}^{t} = \bar{\boldsymbol{Q}}^{H}\mathrm{Diag}(j\bar{\boldsymbol{\kappa}}^{t})\mathrm{Diag}(\boldsymbol{b})\bar{\boldsymbol{Q}}, ~t\in\{x, y\}, 
        \end{equation}
        with $\bar{\boldsymbol{\kappa}}^{t} = [\bar{\kappa}_{1}^{t}, \ldots, \bar{\kappa}_{L_{0}}^{t}]^{T}\in\mathbb{R}^{L_{0}\times 1}$. 
        The derivations of equations~\eqref{def:corr-factor-grads-wrt-positions} and~\eqref{def:corr-factor-grads-Smat} can be found in Appendix~\ref{appendix:corr-factor-grads}. 
        Besides, it is easy to verify that the gradients of $\mathcal{L}_{\mathcal{S}}$ w.r.t. antenna positions are given by
        \begin{equation}\label{def:log-barrier-gradients}
            \frac{\partial\mathcal{L}_{\mathcal{S}}}{\partial{t}_{n}} = \frac{2N^{-1}{t}_{n}}{
                {t}_{n}^2 - {S}_{t}^2/4
            } + \sum_{\substack{
                1\le i\le {N} \\
                i\neq n
            }}{\frac{
                2N^{-2}({t}_{n} - {t}_{i})
            }{
                \|\boldsymbol{r}_{n} - \boldsymbol{r}_{i}\|_2^2 - \Delta^2
            }}, ~\forall n, 
        \end{equation}
        where $t\in\{x, y\}$. 
        Thus, the gradients of $F(\boldsymbol{x}, \boldsymbol{y})$ w.r.t. $(\boldsymbol{x}, \boldsymbol{y})\in\mathrm{int}(\mathcal{S}_{\text{MA}})$ can be computed as 
        \begin{equation}\label{def:lobpo-grads}
            \nabla_{\boldsymbol{t}}{F} \triangleq \nabla_{\boldsymbol{t}}{\rho_{N}} + \alpha\nabla_{\boldsymbol{t}}{\mathcal{L}_{\mathcal{S}}}\in\mathbb{R}^{N\times 1}, ~\boldsymbol{t}\in\{\boldsymbol{x}, \boldsymbol{y}\}, 
        \end{equation}
        where $\nabla_{\boldsymbol{t}}{\mathcal{L}_{\mathcal{S}}} = [\partial\mathcal{L}_{\mathcal{S}}/\partial{t}_{1}, \ldots, \partial\mathcal{L}_{\mathcal{S}}/\partial{t}_{N}]^{T}\in\mathbb{R}^{N\times 1}$.

        The gradient ascent algorithm for solving problem~\eqref{prob:cebap-log-barrier-relaxed} given $\alpha$ is summarized in Algorithm~\ref{alg:gradient-ascent}. 
        Instead of directly using $\nabla_{\boldsymbol{x}}{F}$ and $\nabla_{\boldsymbol{y}}{F}$ for antenna positions' updates, the gradients are normalized, as shown in line $4$, such that the displacement of antennas can be determined by the step size $\epsilon$. 
        Moreover, a proper value for $\epsilon$ is solved in each iteration by applying the backtracking line search to ensure feasibility of $(\boldsymbol{x}, \boldsymbol{y})$ and increase of the objective $F$. 
        Specifically, based on its initial value $\epsilon_{0}$, $\epsilon$ keeps shrinking by half until it satisfies the following confitions, i.e., 
        \begin{subequations}\label{def:backtracking-conditions}
            \begin{gather}
                (\boldsymbol{x}^{(i)} + \epsilon\boldsymbol{d}_{\boldsymbol{x}}^{(i)}, \boldsymbol{y}^{(i)} + \epsilon\boldsymbol{d}_{\boldsymbol{y}}^{(i)})\in\text{int}(\mathcal{S}_{\text{MA}}), \label{subdef:backtracking-feasible} \\
                F(\boldsymbol{x}^{(i)} + \epsilon\boldsymbol{d}_{\boldsymbol{x}}^{(i)}, \boldsymbol{y}^{(i)} + \epsilon\boldsymbol{d}_{\boldsymbol{y}}^{(i)}) \ge F(\boldsymbol{x}^{(i)}, \boldsymbol{y}^{(i)}) + \eta\epsilon\|\boldsymbol{g}^{(i)}\|_2, \label{subdef:backtracking-increasing}
            \end{gather}
        \end{subequations}
        where equation~\eqref{subdef:backtracking-increasing} is the Armijo-Goldstein condition~\cite{ref:Armijo-gradient-backtracking} with control parameter $\eta\in(0, 1)$.

        {
        \begin{algorithm}[t]
            \begin{minipage}{0.95\linewidth}
            \centering
            \caption{The gradient ascent algorithm for solving problem~\eqref{prob:cebap-log-barrier-relaxed} given $\alpha$. }
            \begin{algorithmic}[1]\label{alg:gradient-ascent}
                \REQUIRE Discretized wavevectors $\bar{\boldsymbol{\kappa}}_{l}$, $1\le l\le L_{0}$, APS $\boldsymbol{b}$, initial antenna positions $(\boldsymbol{x}^{(0)}, \boldsymbol{y}^{(0)})$, penalty weight $\alpha$, initial step size $\epsilon_{0}$, and maximum iteration number $I$. 
                \STATE Let $i\gets 0$. 
                \WHILE{$i < I$}
                    \STATE For $(\boldsymbol{x}^{(i)}, \boldsymbol{y}^{(i)})$, compute $\nabla_{\boldsymbol{x}}{F}$ and $\nabla_{\boldsymbol{y}}{F}$ from~\eqref{def:lobpo-grads}. 
                    \STATE Compute $\boldsymbol{g}^{(i)}\gets [\nabla_{\boldsymbol{x}}{F}^{T}, \nabla_{\boldsymbol{y}}{F}^{T}]^{T}$ and normalized gradients $\boldsymbol{d}_{\boldsymbol{t}}^{(i)}\gets \nabla_{\boldsymbol{t}}{F}/\|\boldsymbol{g}^{(i)}\|_2$, $\boldsymbol{t}\in\{\boldsymbol{x}, \boldsymbol{y}\}$. 
                    \STATE Initial step size ${\epsilon}\gets{\epsilon}_{0}$ and keep shrinking the step size as $\epsilon\gets\epsilon/2$ until conditions~\eqref{def:backtracking-conditions} are met. 
                    \STATE Let $\boldsymbol{x}^{(i + 1)}\gets\boldsymbol{x}^{(i)} + \epsilon\boldsymbol{d}_{\boldsymbol{x}}^{(i)}$, $\boldsymbol{y}^{(i + 1)}\gets\boldsymbol{y}^{(i)} + \epsilon\boldsymbol{d}_{\boldsymbol{y}}^{(i)}$, and $i\gets i + 1$. 
                \ENDWHILE
                \RETURN The optimized antenna positions $(\boldsymbol{x}^{(I)}, \boldsymbol{y}^{(I)})$. 
            \end{algorithmic}
            \end{minipage}
        \end{algorithm}
        }

    \subsection{LOBPO Method}\label{subsec:lobpo}
        Applying the gradient ascent algorithm, the proposed LOBPO method gradually approaches a suboptimal solution $(\boldsymbol{x}^{\star}, \boldsymbol{y}^{\star})$ for problem~\eqref{prob:cebap}, i.e., the CEBAP, by iteratively shrinking the penalty weight $\alpha$ with factor $\tau$. 
        Specifically, we start from a relatively large initialization $\alpha_{0}$ and Algorithm~\ref{alg:gradient-ascent} is leveraged to solve problem~\eqref{prob:cebap-log-barrier-relaxed} given each $\alpha$, while the antenna positions solved for the previous penalty weight serve as a good initialization for the next iteration. 
        The LOBPO terminates if the displacement between antenna positions solved for two consecutive penalty weights is negligible, i.e., becomes smaller than a predefined threshold $\varepsilon_{0}$, indicating that the penalty weight is sufficiently small and the solved antenna positions can be approximately regarded as the CEBAP, as summarized in Algorithm~\ref{alg:lobpo}. 

        {
        \begin{algorithm}[t]
            \begin{minipage}{0.95\linewidth}
            \centering
            \caption{LOBPO method for solving the CEBAP. }
            \begin{algorithmic}[1]\label{alg:lobpo}
                \REQUIRE Discretized wavevectors $\bar{\boldsymbol{\kappa}}_{l}$, $1\le l\le L_{0}$, APS $\boldsymbol{b}$, initial penalty weight $\alpha_{0}$, factor $\tau$, initial step size $\epsilon_{0}$, maximum iteration number $I$, and threshold $\varepsilon_{0}$. 
                \STATE Sparse uniform planar array (UPA) is applied as initialization $(\hat{\boldsymbol{x}}^{(0)}, \hat{\boldsymbol{y}}^{(0)})$ (specified in Section~\ref{subsec:benchmarks-and-proposed} as UPA-sparse); $\alpha\gets\alpha_{0}$, $\varepsilon\gets\varepsilon_{0}$, and $i\gets 0$. 
                \WHILE{$\varepsilon \ge \varepsilon_{0}$}
                    \STATE Given initialization $(\hat{\boldsymbol{x}}^{(i)}, \hat{\boldsymbol{y}}^{(i)})$, apply Algorithm~\ref{alg:gradient-ascent} to solve problem~\eqref{prob:cebap-log-barrier-relaxed} and denote the optimized antenna positions as $(\hat{\boldsymbol{x}}^{(i + 1)}, \hat{\boldsymbol{y}}^{(i + 1)})$. 
                    \STATE Let $\varepsilon\gets(\|\hat{\boldsymbol{x}}^{(i + 1)} - \hat{\boldsymbol{x}}^{(i)}\|_{2}^{2} + \|\hat{\boldsymbol{y}}^{(i + 1)} - \hat{\boldsymbol{y}}^{(i)}\|_{2}^{2})^{1/2}$. 
                    \STATE Let $i\gets i + 1$ and $\alpha\gets\tau\alpha$. 
                \ENDWHILE
                \RETURN The optimized antenna positions $(\hat{\boldsymbol{x}}^{(i)}, \hat{\boldsymbol{y}}^{(i)})$. 
            \end{algorithmic}
            \end{minipage}
        \end{algorithm}
        }
        
        It is noteworthy that the LOBPO method is a general framework for antenna position optimization, which is not limited to solving CEBAP but also can be applied to solve the original two-timescale problem~\eqref{prob:two-timescale-optm}. 
        Specifically, by applying suboptimal precoding algorithms catering to the utility function and leveraging Monte-Carlo (MC) simulations, the objective, i.e., the ergodic utility with optimized precoding, can be numerically obtained. 
        Then, the gradients of the objective w.r.t. antenna positions can be approximated via finite difference methods~\cite{ref:finite-difference-method}, based on which the LOBPO method can be employed to optimize the antenna positions. 
        Nevertheless, the computational complexity will be much higher than that of solving CEBAP.


\section{Performance Evaluation}\label{sec:perf-eval-raytrace}

    \subsection{Simulation Setups}\label{subsec:simulation-setups}
        
        \begin{figure}[t!]
            \centering
            {
                \begin{subfigure}[t]{0.3\textwidth}
                    \centering
                    \includegraphics[scale = 0.25]{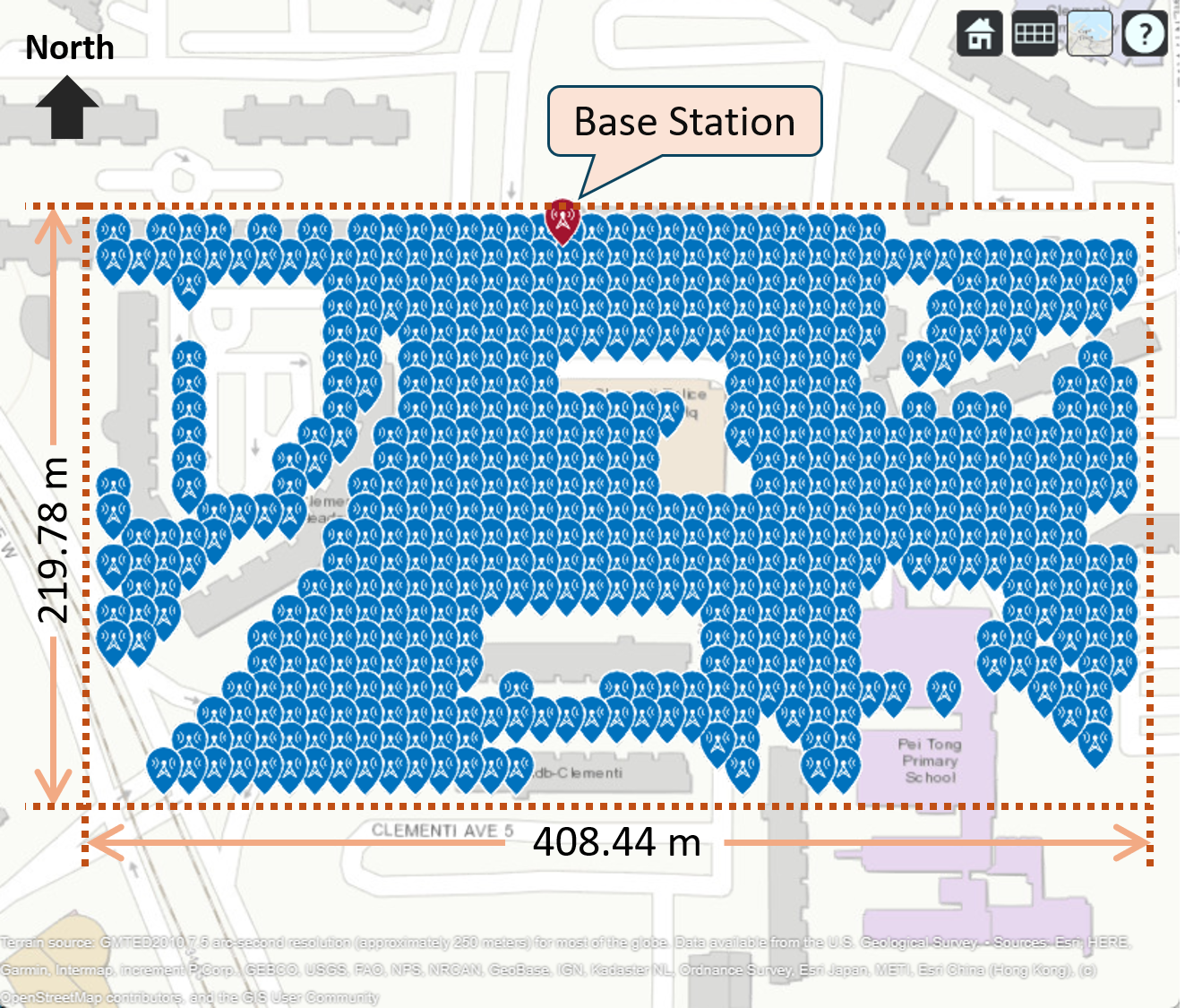}
                    \caption{Subregions division in the cell region. }
                    \label{subfig:cell-subregions}
                \end{subfigure}
                \begin{subfigure}[t]{0.18\textwidth}
                    \centering
                    \includegraphics[scale = 0.25]{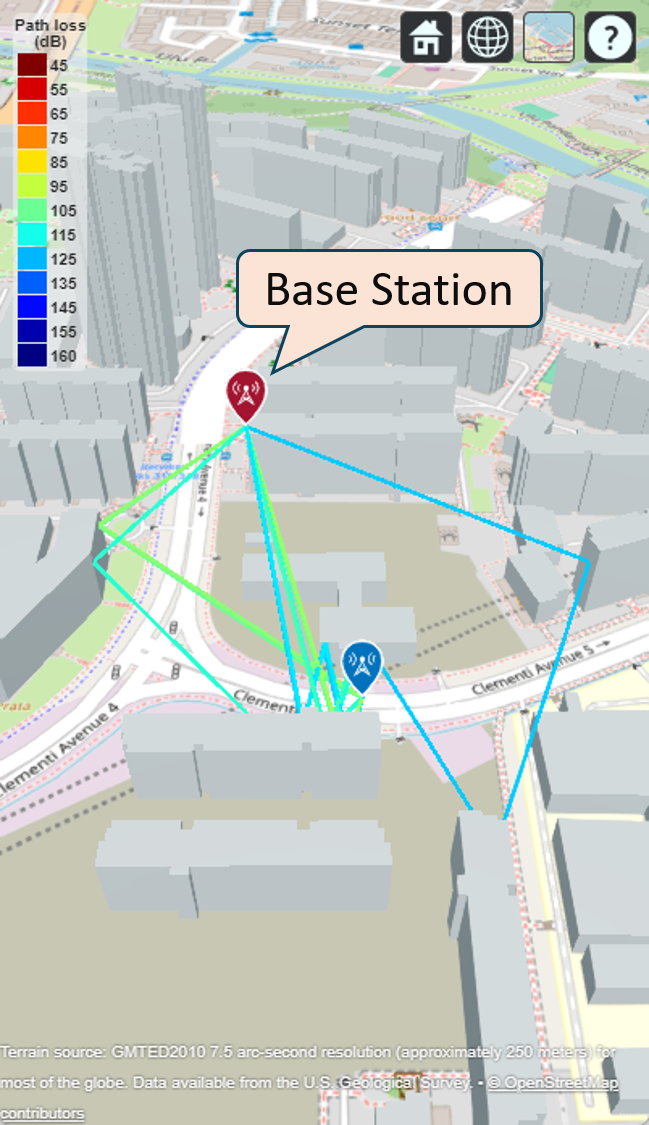}
                    \caption{Ray-tracing channels. }
                    \label{subfig:cell-raytrace}
                \end{subfigure}
            }
            \vspace{-2pt}
            \caption{Site environment setup with $584$ subregions. }
            \label{fig:cell-setup}
            \vspace{-2pt}
        \end{figure}
        
        To validate the effecacy of the proposed CEBAP and LOBPO methods in a practical propagation environment, we conduct simulations with ray-tracing generated channels based on the realistic urban data at Clementi, Singapore~\cite{ref:openstreetmap}. 
        As shown in Fig.~\ref{fig:cell-setup}, a rectangular region with a length of $408.44$ meters (m) and a width of $219.78$ m is considered for the cell. 
        The BS, marked in red, is located above the northmost building with height $59.97$ m, equipped with $N = 4\times 4 = 16$ MAs on the antenna plane facing to the south. 
        The carrier frequency is $f_{c} = 5$ GHz with a wavelength of $\lambda = 6$ cm and the maximum number of reflections is set as $5$ for ray-tracing. 
        By considering only ground users, the cell region is divided into $41\times 22 = 902$ square areas, each of size $10$ m $\times 10$ m, where $M = 584$ of them are adopted as subregions for user channel generation with their centers marked in blue in Fig.~\ref{subfig:cell-subregions}, while the remaining areas are discarded because no channel path is available between them and the BS with less than $5$ reflections. 
        The transmit AoDs and power responses of each subregion are generated as those from the BS to its center, as shown in the example in Fig.~\ref{subfig:cell-raytrace}. 
        In particular, the LoS and NLoS channel paths are scaled with different factors such that the Rician factor across the cell equals to a predefined value $\chi$. 
        For APS representation, the transmit elevation and azimuth AoDs are discretized into $N_{E} = 50$ and $N_{A} = 80$ angles, respectively, yielding $L_{0} = 4000$ solid angle grids. 

        To evaluate the system performance, the ergodic utility is obtained via MC simulations over $5000$ random MU-MISO channel realizations. 
        Specifically, the truncated Poisson distribution is assumed for the user number, i.e., $\zeta_{n} = {Z}^{-1}K_{0}^{n}/n!$, $1\le n\le N$, where $K_{0}$ is the Poisson parameter and $Z$ is the normalization factor that ensures $\sum_{n = 1}^{N}{\zeta_{n}} = 1$. 
        For each channel realization, the user number $K$ is first generated, based on which $K$ subregions are selected according to the user distribution across the cell, i.e., $\mu_{1}, \ldots, \mu_{M}$. 
        After that, the Gaussian path-response coefficients are generated for each transmit channel path of the subregions. 
        Depending on the utility function, different precoding algorithms are employed given the instantaneous channel to compute the instantaneous utility, e.g., the reduced weighted minimum mean square error (RWMMSE) precoding~\cite{ref:RWMMSE} and the max-min weighted SINR precoding~\cite{ref:max-min-weighted-sinr-precoding} for maximizing the weighted sum rate and minimum weighted SINR, respectively. 
        
        Unless otherwise stated, we set the transmit power as $P_{T} = 30$ dBm, noise power as $\sigma^2 = -90$ dBm, $\Delta = \lambda/2$, ${S}_{x} = {S}_{y} = {S}_{0} = 4\lambda$, $K_{0} = 12$, and $\chi = 10$ dB. 
        Parameters for Algorithms~\ref{alg:gradient-ascent} and~\ref{alg:lobpo} are set as $\alpha_{0} = 1$, $\epsilon_{0} = 0.2\lambda$, $\varepsilon_{0} = 0.01\lambda$, $\eta = 10^{-4}$, $I_{\text{c}} = 20$, $I = 25$, and $\tau = 0.2$. 

    \subsection{Benchmarks for Performance Comparison}\label{subsec:benchmarks-and-proposed}
        The proposed CEBAP solved by the LOBPO method is evaluated and compared with the following benchmarks: i) \textbf{UPA-dense}: The conventional dense UPA is employed, where the antennas forms a $4\times 4$ array with inter-antenna separation $\lambda/2$. ii) \textbf{UPA-sparse}: The antennas are sparsely placed into a $4\times 4$ array with inter-antenna separations ${S}_{x}/4$ and ${S}_{y}/4$ along the $x$ and $y$ axes, respectively. iii) \textbf{MA, cell-specific}: The antenna positions are optimized for the original problem~\eqref{prob:two-timescale-optm} by levaraging the LOBPO method, where the gradients of the objective are computed via finite difference methods and MC simulations, as discussed in Section~\ref{subsec:lobpo}. iv) \textbf{MA, subregion-specific}: The antenna positions are optimized based on the subregion-specific (i.e., user-location-specifc) statistical CSI~\cite{ref:my-ma-tcom} by the LOBPO method, which are updated once users move to other subregions. The gradients are computed via finite difference methods and MC simulations. v) \textbf{MA, instantaneous}: The antenna positions are optimized based on instantaneous CSI by the LOBPO method, with the gradients computed via finite difference methods. 

        Note that three MA-based benchmarks optimize antenna positions over different timescales, where ``MA, cell-specific'' can be regarded as a computationally heavy baseline for the proposed MA optimization approach, while the other two serve as upper bounds. 
        To solve antenna positions for ``MA, cell-specific'' and ``MA, subregion-specific'', we adopt $500$ and $25$ random MU-MISO channel realizations in each MC simulation for the LOBPO method, respectively. 

    \subsection{Convergence and Effectiveness of Proposed CEBAP}\label{subsec:alg-convergence}
        Assuming uniform user distribution over $M$ subregions, i.e., $\mu_{m} = 1 / M, \forall m$, the proposed CEBAP is solved by the LOBPO method for the considered cell region, which is visualized in Fig.~\ref{fig:cebap-visualization}. 
        In particular, the cell-specific APSD is shown in Fig.~\ref{subfig:env-apsd}, where the total average channel power gain is $\beta = -45.51$ dB. 
        It can be observed that the channel power is distributed over the lower-half of $\mathcal{S}_{+}$, i.e., where $\kappa^{y} < 0$, because the BS stands higher than all users and scatterers. 
        Besides, more power is concentrated to the west due to buildings there, which frequently reflect signals from the BS to users across the cell.

        \begin{figure*}[t!]
            \centering
            {
                \begin{subfigure}[t]{0.35\textwidth}
                    \centering
                    \includegraphics[scale = 0.49]{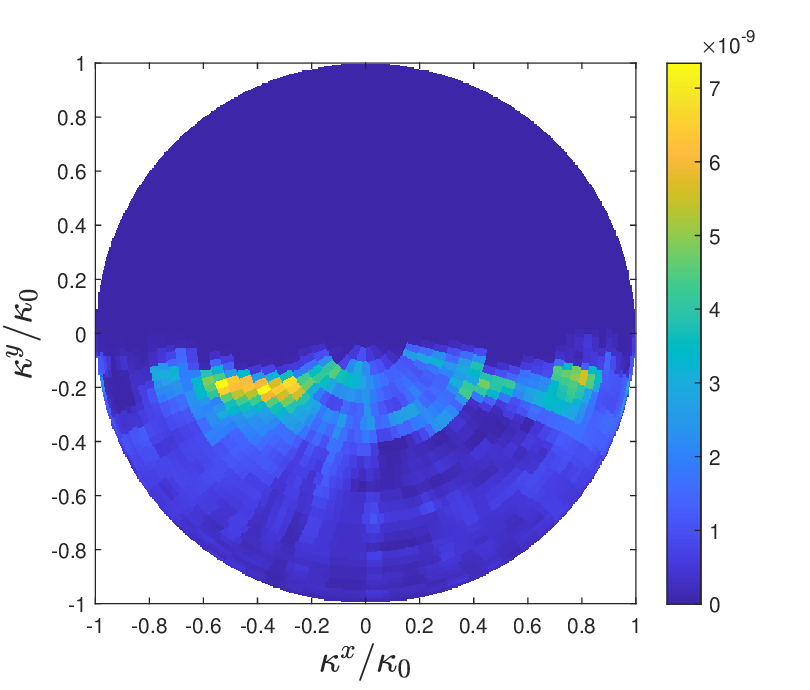}
                    \caption{Transmit APSD at the BS. }
                    \label{subfig:env-apsd}
                \end{subfigure}
                \begin{subfigure}[t]{0.33\textwidth}
                    \centering
                    \includegraphics[scale = 0.49]{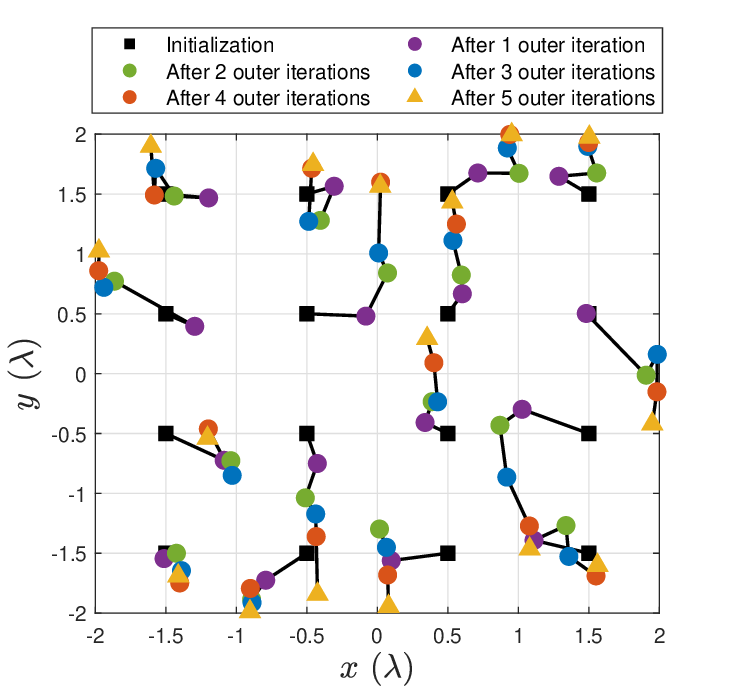}
                    \caption{Trajectories of MAs during iterations. }
                    \label{subfig:cebap-traj}
                \end{subfigure}
                \begin{subfigure}[t]{0.3\textwidth}
                    \centering
                    \includegraphics[scale = 0.49]{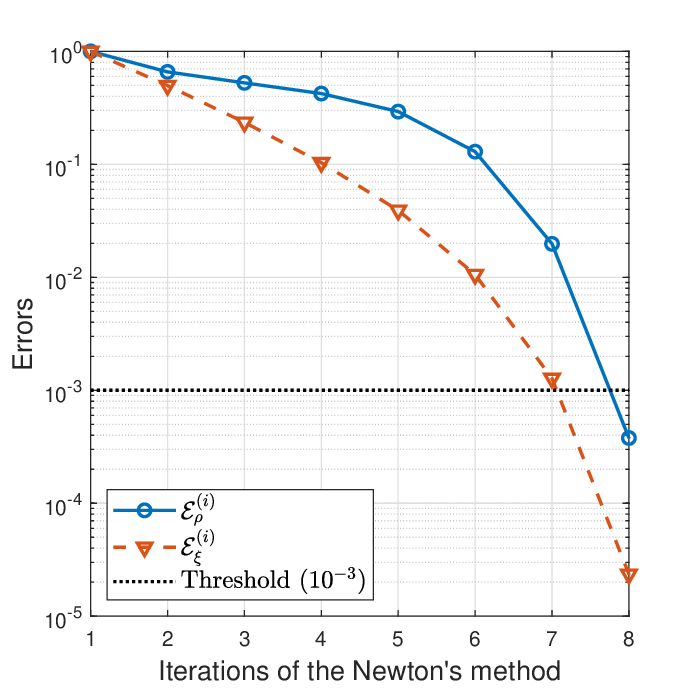}
                    \caption{Convergence of the Newton's method. }
                    \label{subfig:cebap-newton-conv}
                \end{subfigure}

                \begin{subfigure}[t]{0.26\textwidth}
                    \centering
                    \includegraphics[scale = 0.48]{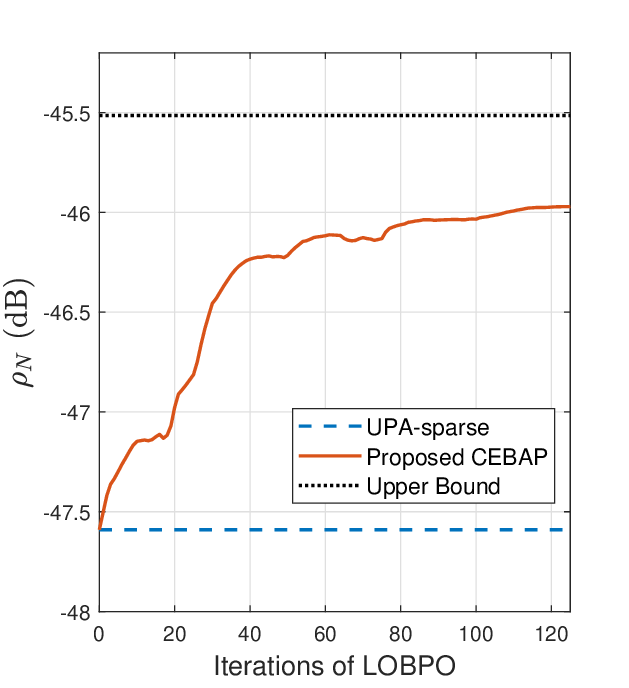}
                    \caption{Convergence of LOBPO method. }
                    \label{subfig:cebap-lobpo-conv}
                \end{subfigure}
                \begin{subfigure}[t]{0.245\textwidth}
                    \centering
                    \includegraphics[scale = 0.48]{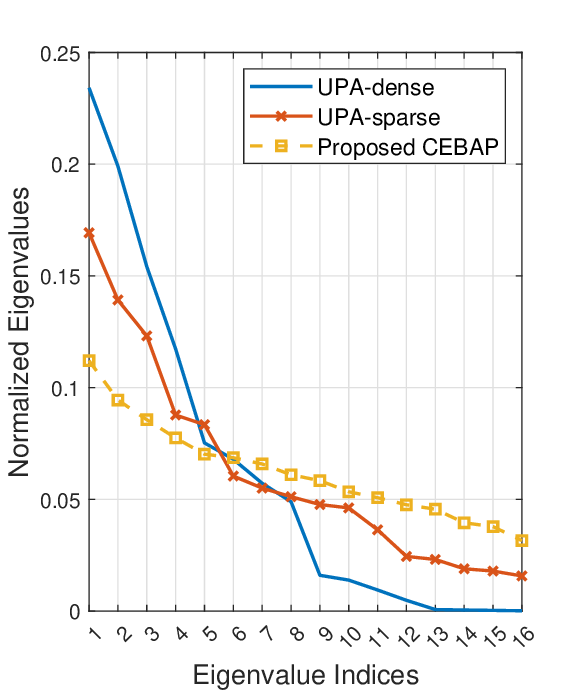}
                    \caption{Eigenvalues of $\bar{\boldsymbol{G}}$. }
                    \label{subfig:cov-eigens}
                \end{subfigure}
                \begin{subfigure}[t]{0.46\textwidth}
                    \centering
                    \includegraphics[scale = 0.46]{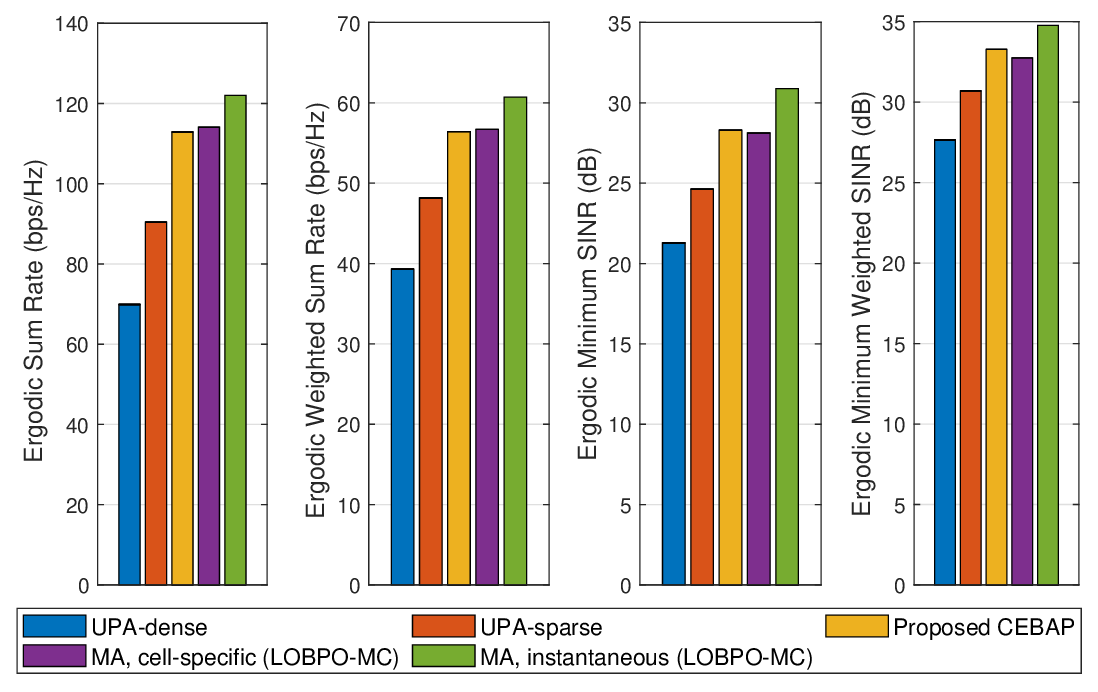}
                    \caption{Utility comparisons of CEBAP with benchmarks. }
                    \label{subfig:cebap-perf-bar}
                \end{subfigure}
            }
            \caption{Visualization of MA optimization based on the proposed CEBAP. }
            \label{fig:cebap-visualization}
            \vspace{-8pt}
        \end{figure*}

        Based on this APS, MAs at the BS are optimized to maximize $\rho_{N}$ and their trajectories during the outer iterations of the LOBPO method (iterations of Algorithm~\ref{alg:lobpo}) are presented in Fig.~\ref{subfig:cebap-traj}, converging to the CEBAP marked with yellow triangles. 
        To verify the convergence of the Newton's method employed in Section~\ref{subsec:asymp-approx}, the asymptotic decorrelated channel power gain $\rho_{N}$ is computed given the solved CEBAP. 
        Specifically, by denoting $\rho_{N}^{(i)}$ as the asymptotic decorrelated channel power gain obtained in the $i$-th iteration, we define $\mathcal{E}_{\rho}^{(i)} = |\rho_{N}^{(i)} - \rho_{N}^{(i - 1)}|/|\rho_{N}^{(i)}|$ and $\mathcal{E}_{\xi}^{(i)} = |\xi_{N}(\rho_{N}^{(i)}) - 1|$ as the relative error of $\rho_{N}^{(i)}$ with the previous iteration and the error of equation $\xi_{N}(\rho_{N}) = 1$, respectively. 
        As is shown in Fig.~\ref{subfig:cebap-newton-conv}, both $\mathcal{E}_{\rho}^{(i)}$ and $\mathcal{E}_{\xi}^{(i)}$ decrease rapidly, which indicates only $8$ iterations are required to solve $\rho_{N}$ with an error no more than $10^{-3}$. 
        Meanwhile, the convergence of the LOBPO method for solving the CEBAP is shown in Fig.~\ref{subfig:cebap-lobpo-conv}, where the asymptotic decorrelated channel power gain is solved for antenna positions updated in every iteration of the LOBPO method. 
        Due to the log-barrier penalties, $\rho_{N}$ is not strictly decreasing during iterations but eventually converges closely to the upper bound $\rho_{N}^{\text{max}} = \beta = -45.51$ dB, with $1.62$ dB improvement from that of UPA-sparse. 

        In Fig.~\ref{subfig:cov-eigens}, the normalized eigenvalues of the channel covariance matrices are shown for FPAs and CEBAP. 
        The eigenvalue distribution under CEBAP is significantly more balanced than that of FPAs, indicating reduced spatial correlation among user channels.
        Furthermore, the performance of the obtained CEBAP is evaluated under various utility functions and compared with benchmark schemes in Fig.~\ref{subfig:cebap-perf-bar}, where the weights for the weighted sum rate and minimum weighted SINR are randomly generated. 
        Notably, the benchmark scheme ``MA, cell-specific'' optimizes antenna positions separately for each utility, whereas CEBAP remains fixed.
        The results demonstrate that CEBAP consistently improves multiple performance metrics, achieving performance close to that of ``MA, instantaneous'' which optimizes antenna positions based on instantaneous CSI. 
        In addition, CEBAP even outperforms antenna positions optimized via MC simulations w.r.t. ergodic minimum SINR and ergodic minimum weighted SINR. 
        This is caused by the limited number of MC samples used in the benchmark optimization.

        \begin{figure*}[t!]
            \centering
            {
                \begin{subfigure}[t]{0.325\textwidth}
                    \centering
                    \includegraphics[scale = 0.42]{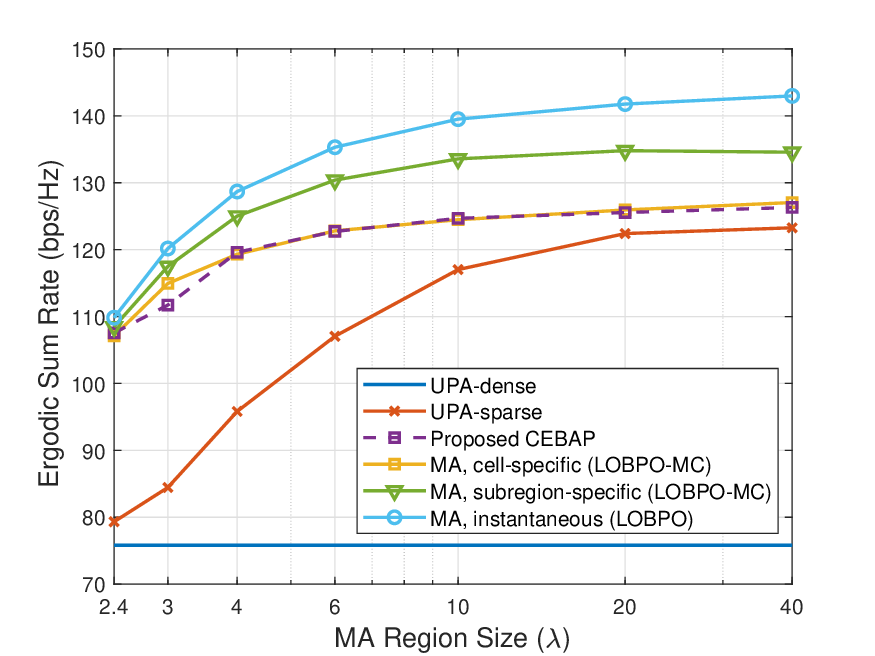}
                    \caption{Ergodic sum rate versus MA region size. }
                    \label{subfig:per-sys-aper-sr}
                \end{subfigure}
                \begin{subfigure}[t]{0.325\textwidth}
                    \centering
                    \includegraphics[scale = 0.42]{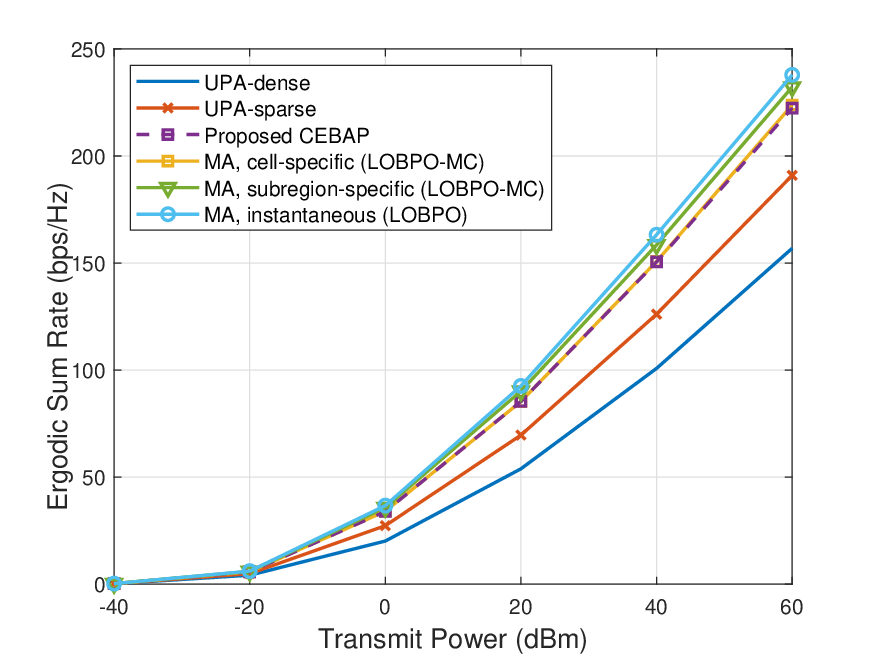}
                    \caption{Ergodic sum rate versus transmit power. }
                    \label{subfig:per-sys-txpwr-sr}
                \end{subfigure}
                \begin{subfigure}[t]{0.325\textwidth}
                    \centering
                    \includegraphics[scale = 0.42]{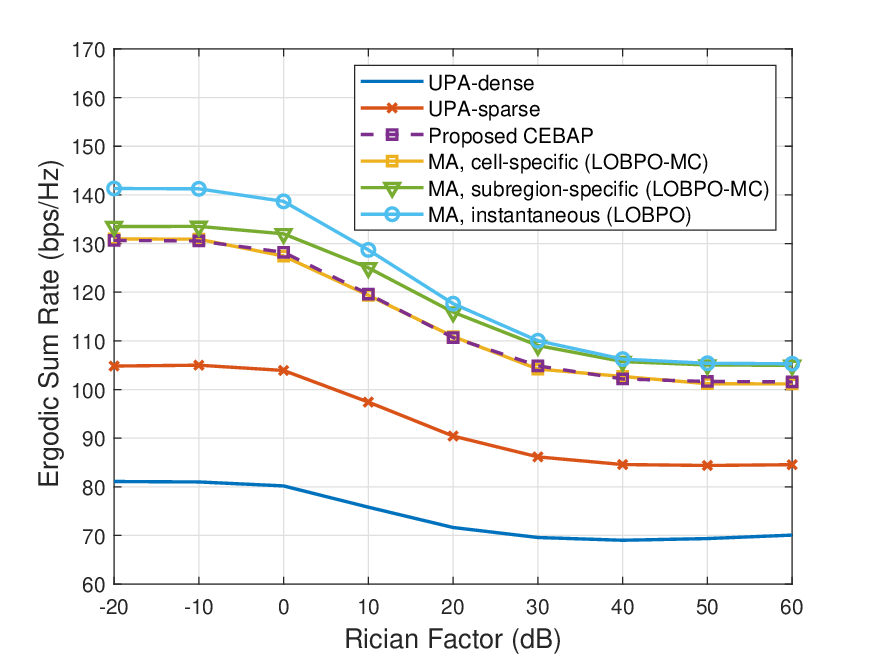}
                    \caption{Ergodic sum rate versus Rician factor. }
                    \label{subfig:per-sys-rician-sr}
                \end{subfigure}

                \begin{subfigure}[t]{0.325\textwidth}
                    \centering
                    \includegraphics[scale = 0.42]{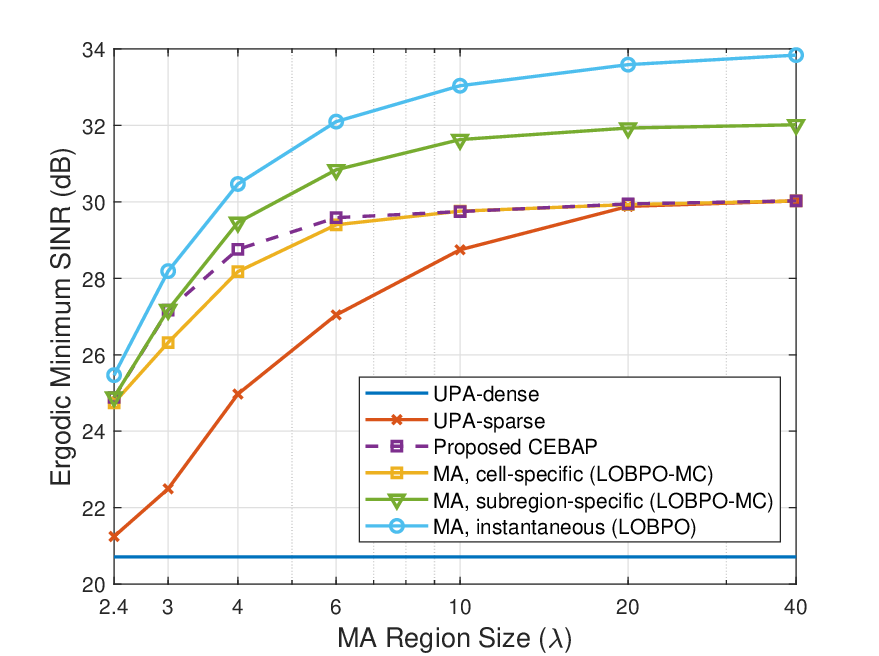}
                    \caption{Ergodic minimum SINR versus MA region size. }
                    \label{subfig:per-sys-aper-msinr}
                \end{subfigure}
                \begin{subfigure}[t]{0.325\textwidth}
                    \centering
                    \includegraphics[scale = 0.42]{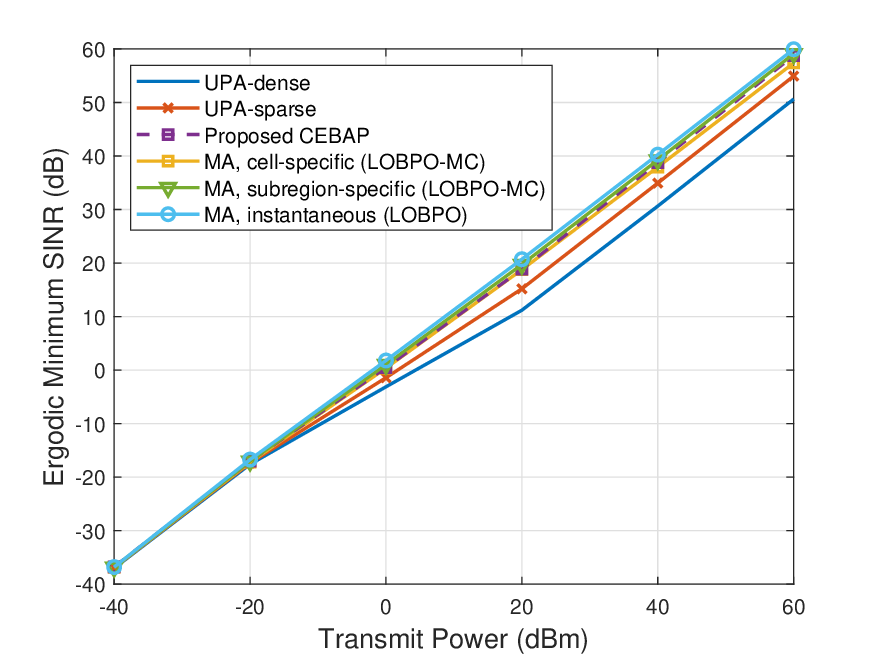}
                    \caption{Ergodic minimum SINR versus transmit power. }
                    \label{subfig:per-sys-txpwr-msinr}
                \end{subfigure}
                \begin{subfigure}[t]{0.325\textwidth}
                    \centering
                    \includegraphics[scale = 0.42]{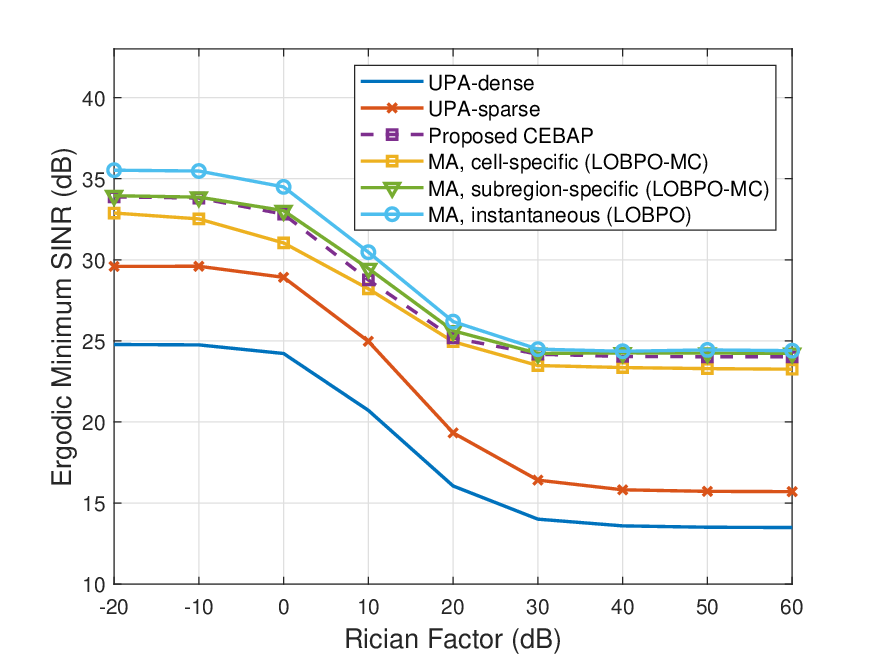}
                    \caption{Ergodic minimum SINR versus Rician factor. }
                    \label{subfig:per-sys-rician-msinr}
                \end{subfigure}
            }
            \vspace{-2pt}
            \caption{Ergodic sum rate and minimum SINR versus different system configurations. }
            \label{fig:perf-sys}
            \vspace{-12pt}
        \end{figure*}

    \subsection{Utility Improvement versus System Configurations}\label{subsec:utility-max-versus-env}
        In this subsection, the performance improvements of the CEBAP are compared with benchmarks under various system configurations. 
        As shown in Fig.~\ref{fig:perf-sys}, all MA designs as well as FPAs are evaluated with different MA region size, transmit power, and channel Rician factor. 
        The MA region size is the side length $S_{0}$ normalized by wavelength $\lambda$, which equals $\Delta\times 4/\lambda = 2$ for the $4\times 4$ UPA-dense array. 
        The results demonstrate significant performance boost for the CEBAP over FPA systems, which even exceed that of ``MA, cell-specific'' in terms of ergodic minimum SINR, similar to Fig.~\ref{subfig:cebap-perf-bar}. 
        
        From Fig.~\ref{subfig:per-sys-aper-sr} and~\ref{subfig:per-sys-aper-msinr}, it can be observed that the ergodic system performance of MA designs based on cell-specific information (i.e., CEBAP and ``MA, cell-specific'') coincide with that of UPA-sparse when the MA region is sufficiently large. 
        This verifies the analysis in Section~\ref{subsec:discussions} that arbitrary antenna position design obtains the same ergodic performance as inter-antenna spacing goes to infinity. 
        Meanwhile, it is also validated that the most significant gain can be achieved for moderate MA region size, e.g., $24.73\%$ and $3.79$ dB improvements for ergodic sum rate and minimum SINR when $S_{0} = 4\lambda$, respectively. 
        Moreover, despite relying on ZF precoding and assuming a large number of transmit channel paths, consistent performance gain is obtained by CEBAP over FPA systems for arbitrary transmit power and Rician factor, as shown in Figs.~\ref{subfig:per-sys-txpwr-sr},~\ref{subfig:per-sys-txpwr-msinr},~\ref{subfig:per-sys-rician-sr}, and~\ref{subfig:per-sys-rician-msinr}. 
        Since the transmit power does not affect the user channel correlation, the performance differences between any two schemes for both ergodic sum rate and minimum SINR become constant as $P_{T}\to\infty$. 
        Meanwhile, the channel correlation among users becomes larger for a large Rician factor $\chi$ even though the total averaged channel power remains unchanged, leading to lower performance as $\chi\to\infty$. 
        By balancing eigenvalues of the channel covariance matrix, the CEBAP is shown capable of decoupling user channels in LoS-dominant cases as well.

    \subsection{Utility Improvement versus User Distribution}\label{subsec:utility-max-versus-userpdf}
        Finally, we investigate the performance of the proposed approach under different user distributions. 
        In Figs.~\ref{subfig:perf-ue-usernum-sr} and~\ref{subfig:perf-ue-usernum-msinr}, the ergodic sum rate and minimum SINR given antenna positions for all schemes are presented with increasing user number Poisson parameter $K_{0}$ assuming $\mu_{m} = 1/M$, $\forall m$. 
        As $K_{0}$ increases, the growth speeds of ergodic sum rates for FPAs become lower due to higher user channel correlation, whereas the CEBAP keeps increasing at a higher speed by reducing the correlation, achieving nearly the same performance as ``MA, cell-specific'' and closely approaching the upper bounds. 
        Similarly, significant gains are obtained for ergodic minimum SINR by CEBAP over FPAs, which is even larger with more users.

        \begin{figure}
            \centering
            \vspace{-4pt}
            {
                \begin{subfigure}[t]{0.4\textwidth}
                    \centering
                    \includegraphics[scale = 0.47]{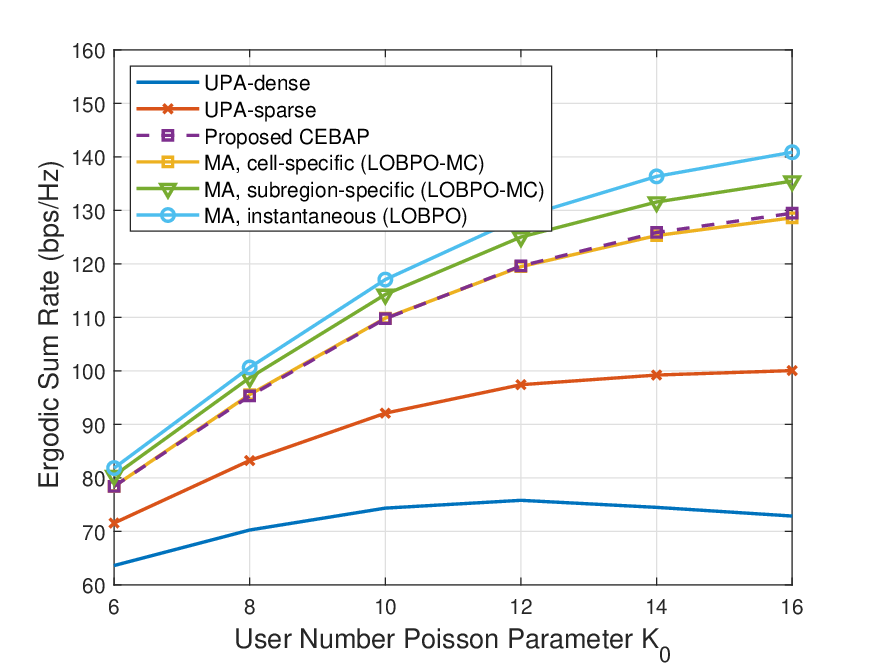}
                    \caption{Ergodic sum rate versus $K_{0}$. }
                    \label{subfig:perf-ue-usernum-sr}
                \end{subfigure}
                \begin{subfigure}[t]{0.4\textwidth}
                    \centering
                    \includegraphics[scale = 0.47]{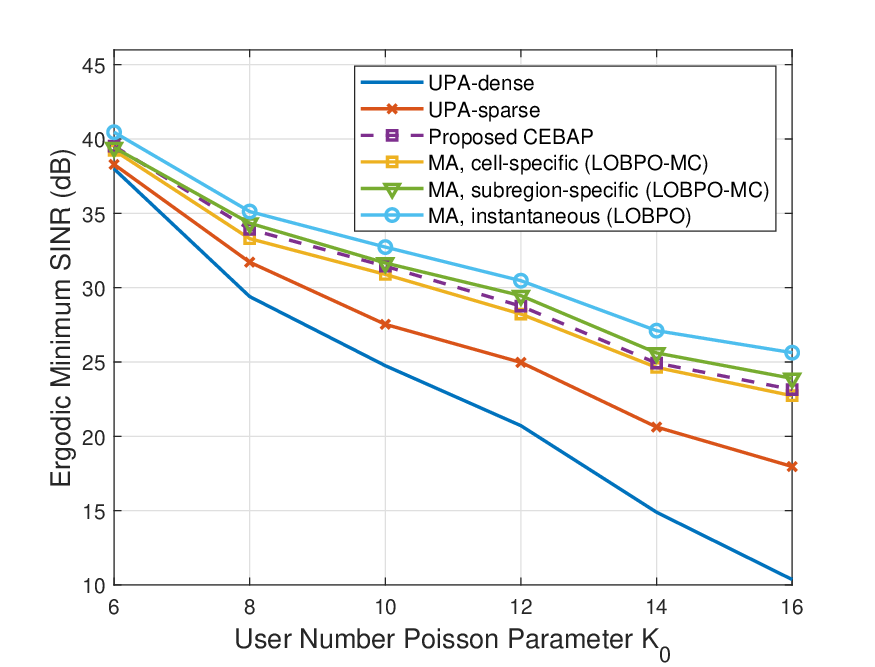}
                    \caption{Ergodic minimum SINR versus $K_{0}$. }
                    \label{subfig:perf-ue-usernum-msinr}
                \end{subfigure}
            }
            \vspace{-2pt}
            \caption{Ergodic sum rate and minimum SINR versus different user number Poisson parameter $K_{0}$. }
            \label{fig:perf-ue-usernum}
            \vspace{-2pt}
        \end{figure}

        Furthermore, we consider a Gaussian user distribution, with its center traversing the cell from the southwest corner to the northeast corner, which is denoted by a 2D vector $a_{\text{ue}}\in\mathbb{R}^{2\times 1}$ on the ground whose elements represent its distances in meters from the southwest corner horizontally and vertically, respectively. 
        The trajectory of $a_{\text{ue}}$ is plotted with the black arrows in Fig.~\ref{fig:user-distrbtn-traverse}. 
        Specifically, $\boldsymbol{a}_{\text{ue}}$ is given by 
        \begin{equation}\label{def:ug-center-traverse}
            \boldsymbol{a}_{\text{ue}} = (1 - \eta_{\text{ue}})\boldsymbol{a}_{\text{SW}} + \eta_{\text{ue}}\boldsymbol{a}_{\text{NE}}, 
        \end{equation}
        where $\eta_{\text{ue}}\in[0, 1]$ is the user center traverse factor, while $\boldsymbol{a}_{\text{SW}}$ and $\boldsymbol{a}_{\text{NE}}$ denote the locations of the southwest corner and the northeast corner, respectively. 
        By denoting $\boldsymbol{a}_{m}\in\mathbb{R}^{2\times 1}$ as the center location of the $m$-th subregion, the user distribution across subregions is defined as
        \begin{equation}\label{def:ug-probs}
            \mu_{m} = \frac{1}{A}\exp\left(
                -\frac{\|\boldsymbol{a}_{m} - \boldsymbol{a}_{\text{ue}}\|_{2}^{2}}{2\sigma_{\text{ue}}^{2}}
            \right), ~1\le m\le M, 
        \end{equation}
        where the spread factor $\sigma_{\text{ue}}$ is set as $54.21$ m while $A$ is the normalization factor such that $\sum_{m = 1}^{M}{\mu_{m}} = 1$. 
        
        \begin{figure}[t]
            \begin{center}
                \includegraphics[scale = 0.325]{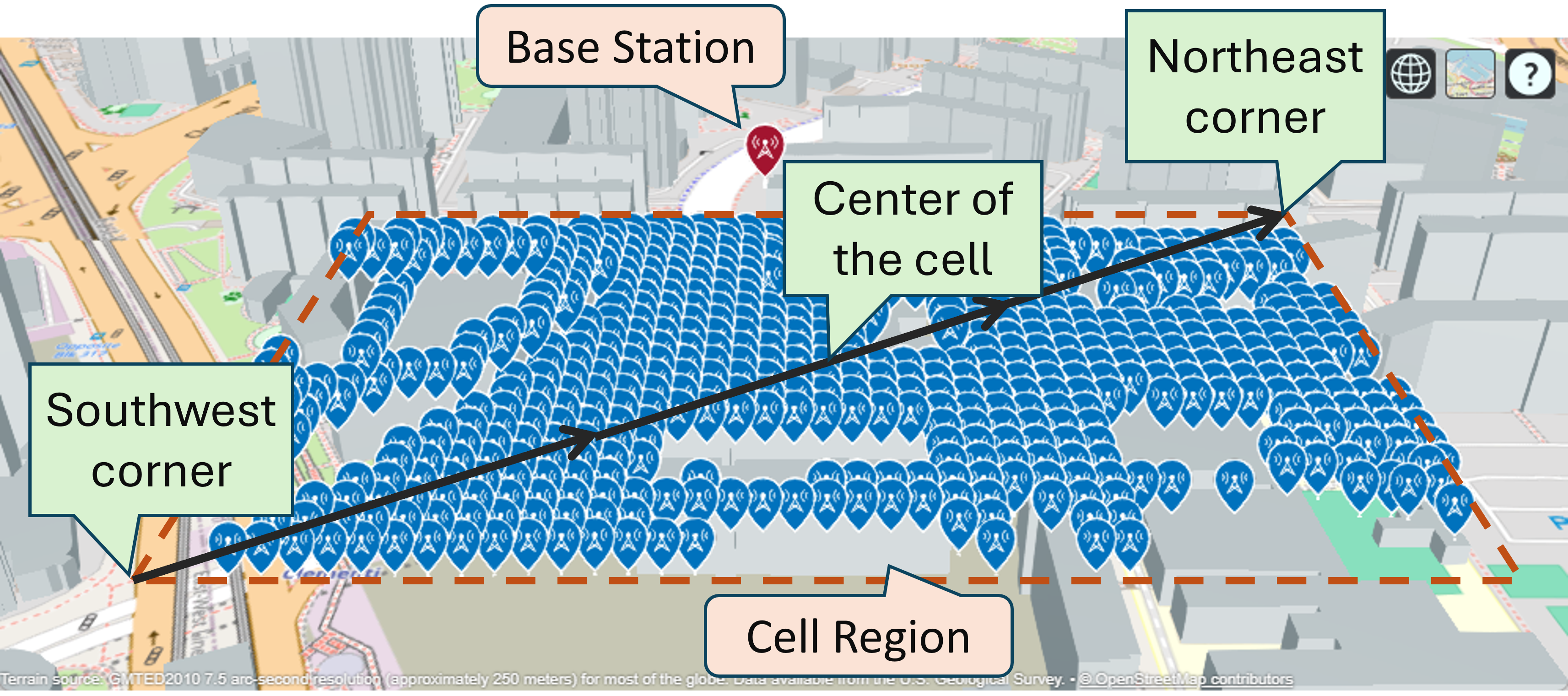}
                \caption{The traverse of the user distribution center. }
                \label{fig:user-distrbtn-traverse}
            \end{center}
        \end{figure}

        \begin{figure}[t!]
            \centering
            {
                \begin{subfigure}[t]{0.24\textwidth}
                    \centering
                    \includegraphics[scale = 0.332]{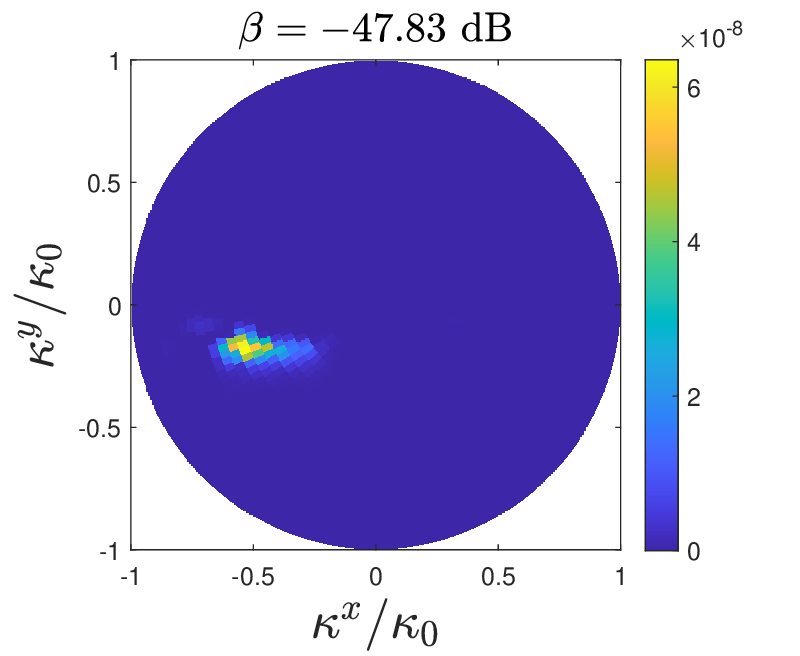}
                    \caption{Southwest corner. }
                    \label{subfig:ug-apsd-southwest}
                \end{subfigure}
                \begin{subfigure}[t]{0.24\textwidth}
                    \centering
                    \includegraphics[scale = 0.332]{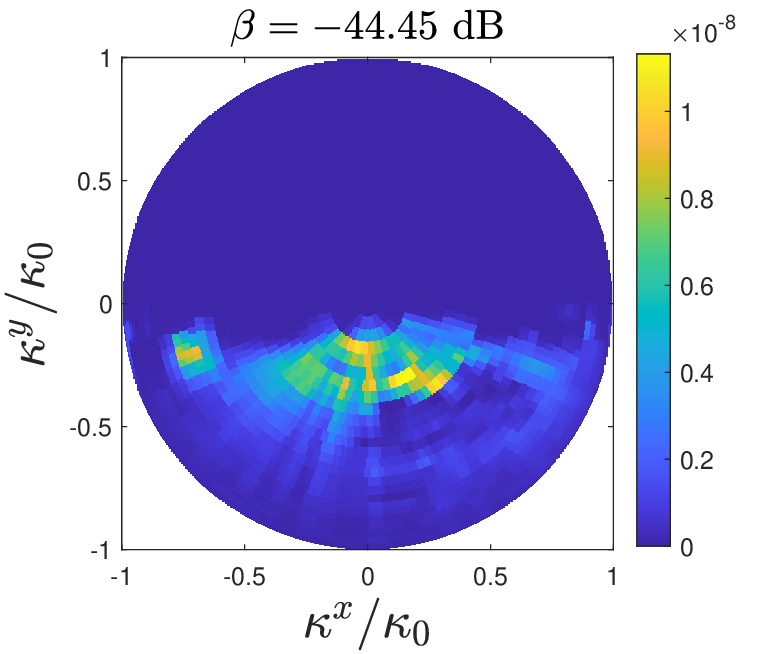}
                    \caption{Center of cell. }
                    \label{subfig:ug-apsd-center}
                \end{subfigure}
                \begin{subfigure}[t]{0.24\textwidth}
                    \centering
                    \includegraphics[scale = 0.332]{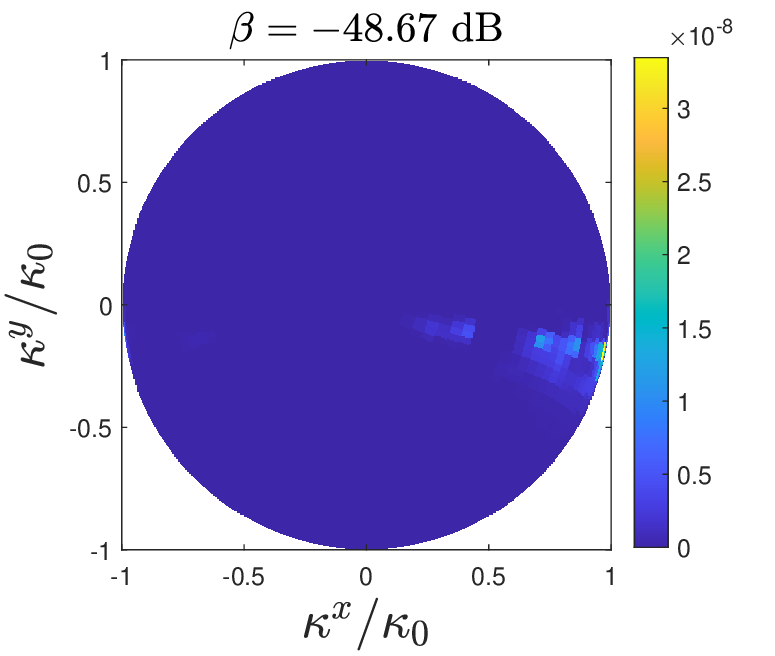}
                    \caption{Northeast corner. }
                    \label{subfig:ug-apsd-northeast}
                \end{subfigure}
            }
            \caption{APSDs given different user distribution centers. }
            \label{fig:ug-apsd}
        \end{figure}

        In Figs.~\ref{subfig:ug-apsd-southwest},~\ref{subfig:ug-apsd-center}, and~\ref{subfig:ug-apsd-northeast}, the APSDs are plotted for $\eta_{\text{ue}} = 0$, $0.5$, and $1$, i.e., the user distribution center is at the southwest corner, center of the cell, and the northeast corner, respectively, where the corresponding total averaged channel power gain is given above the figures. 
        The APSDs for $\eta_{\text{ue}} = 0$ and $\eta_{\text{ue}} = 1$ are spiked and their total power is lower because the majority of the users are distributed at remote corners with their LoS paths blocked or severely attenuated, rendering their channels mainly composed of similar NLoS paths and thus leading to weak and concentrated channel power in the angular domain. 
        In contrast, the total channel power for $\eta_{\text{ue}} = 0.5$ is higher and distributed more uniformly in the angular domain. 

        Based on the APS corresponding to different values of the traverse factor $\eta_{\text{ue}}$, the ergodic sum rate and minimum SINR obtained by CEBAP are shown along with other benchmark schemes in Figs.~\ref{subfig:perf-ue-center-sr} and~\ref{subfig:perf-ue-center-msinr}, respectively. 
        As $\eta_{\text{ue}}$ grows from $0$ to $0.5$, i.e., the user distribution center $\boldsymbol{a}_{\text{ue}}$ moves from the southwest corner to the center of cell, users get closer to the BS and channel power of users is more spread out across the angular domain, which reduce signal attenuation and user channel correlation at the same time, thus leading to performance improvement for both ergodic sum rate and minimum SINR. 
        After that, $\boldsymbol{a}_{\text{ue}}$ moves towards the northeast corner, where distances between the BS and most users become larger and their channels share similar NLoS paths due to obstacles. 
        The system performance decreases during this process because of higher path loss and user channel correlation. 
        Furthermore, the CEBAP always achieves almost the same performance as its upper bounds, exceeding ``MA, cell-specific'' for most cases, which validates its adaptability to different user distributions. 

        \begin{figure}[t!]
            \centering
            \vspace{-4pt}
            {
                \begin{subfigure}[t]{0.4\textwidth}
                    \centering
                    \includegraphics[scale = 0.47]{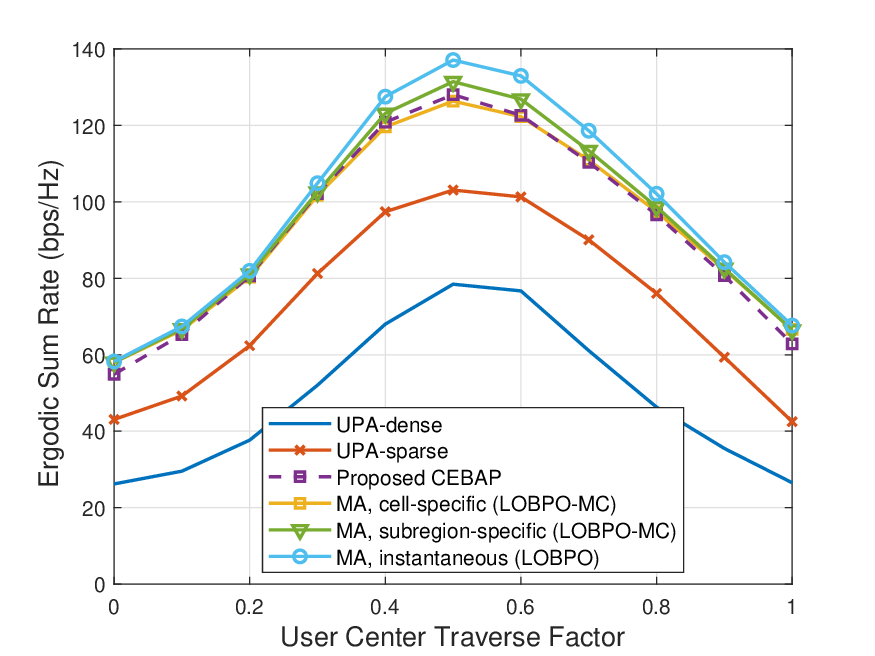}
                    \caption{Ergodic sum rate versus $\eta_{\text{ue}}$. }
                    \label{subfig:perf-ue-center-sr}
                \end{subfigure}
                \begin{subfigure}[t]{0.4\textwidth}
                    \centering
                    \includegraphics[scale = 0.47]{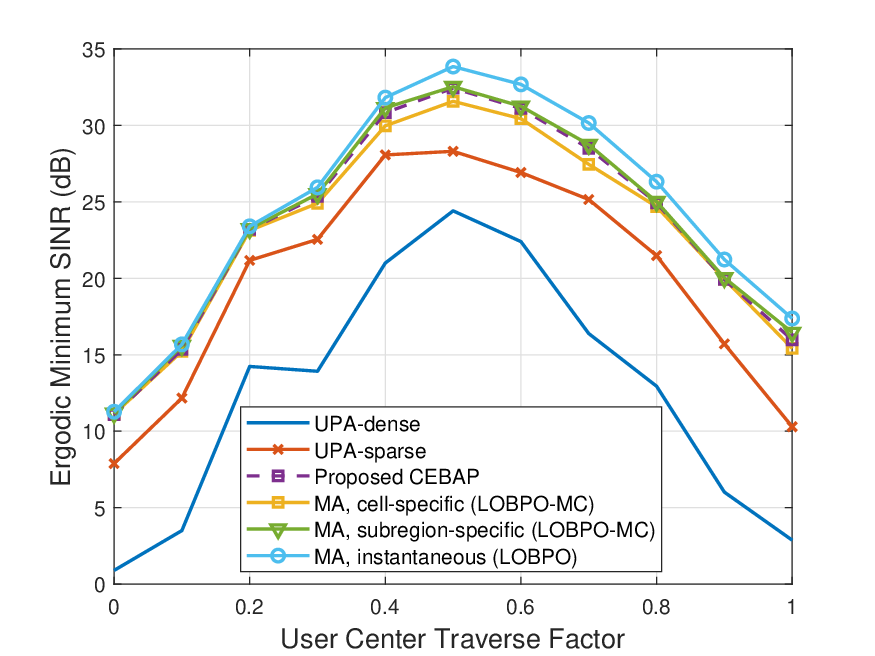}
                    \caption{Ergodic minimum SINR versus $\eta_{\text{ue}}$. }
                    \label{subfig:perf-ue-center-msinr}
                \end{subfigure}
            }
            \caption{Ergodic sum rate and minimum SINR versus user center traverse factor $\eta_{\text{ue}}$. }
            \label{fig:perf-ue-ug-center}
        \end{figure}


\section{Conclusions}\label{sec:conclusions}
    In this paper, we investigated cell-specific long-term antenna position designs for MA-aided MU-MISO downlink transmissions. 
    By optimizing antenna positions based on the APS over an extended timescale, the proposed approach significantly reduces both the antenna movement frequency and channel acquisition overhead, thereby enhancing the practicality of MA deployment at the BS. 
    Specifically, we first developed a cell-specific Gaussian mixture channel model to describe user channel variations over the extended timescale, accounting for both the scattering environment and user distribution. 
    Based on this model, the MA optimization was formulated as a two-timescale optimization problem, where the antenna positions are designed for maximizing the ergodic system utility, while precoding adapts to instantaneous CSI. 
    Using ZF precoding, the CEBAP was derived via asymptotic analysis to approximate the optimal solution. 
    Rather than directly focusing on a specific utility, the proposed CEBAP mitigates inter-user channel correlation by exploiting the APS to equalize eigenvalues of the channel covariance matrix, aiming for a balanced channel eigenspace independent of the utility function. 
    As such, the resulting design is able to enhance system performance in various aspects while avoiding heavy channel acquisition burden. 
    To enable efficient implementation, the LOBPO method was further proposed to numerically obtain the CEBAP solutions through maximizing the asymptotic decorrelated channel power gain incorporated by log-barrier penalties for satisfying practical constraints. 
    Simulation results based on realistic urban environments and ray-tracing channel models validate the effectiveness of the proposed approach across various utility functions.
    In particular, the performance gains are shown to be more pronounced for moderately large MA regions and moderately concentrated APS, as demonstrated by both analytical and simulation results.


\appendices

    \section{Proof for Proposition~\ref{prop:asymp-approx}}
    \label{appendix:asymp-approx-proposition}
        According to~\cite{ref:my-ma-tcom}, constant deterministic equivalents for $c_{k}$ can be obtained under the conditions of Proposition~\ref{prop:asymp-approx}. 
        By defining $\mathcal{M}(k)$ as the index of the subregion in which user $k$ is located, we have $c_{k} - c_{k}^{\infty}\rightarrow 0$ almost surely, where 
        \begin{subequations}
            \begin{gather}
                c_{k}^{\infty} = \left[
                    \mathrm{tr}\left(
                        \boldsymbol{G}_{\mathcal{M}(k)}\boldsymbol{Y}_{k}^{-1}
                    \right)
                \right], \\
                \boldsymbol{Y}_{k} = \boldsymbol{I}_{N} + \sum_{i\neq k}{\epsilon_{k, i}\boldsymbol{G}_{\mathcal{M}(i)}}, 
            \end{gather}
        \end{subequations}
        while $\epsilon_{k, i}$, $1\le i\le K$, are the unique solutions to the following equation, 
        \begin{equation}
            \mathrm{tr}\left(
                \epsilon_{k, i}\boldsymbol{G}_{\mathcal{M}(i)}\boldsymbol{Y}_{k}^{-1}
            \right) = 1, ~1\le k\neq i\le K. 
        \end{equation}
        Note that, for any $k, i$ such that $\mathcal{M}(i) = m$, we have 
        \begin{equation}
            \epsilon_{k, i}^{-1} = \mathrm{tr}\left[
                \boldsymbol{G}_{m}\left(\boldsymbol{I}_{N} + \sum_{l\neq k}{\epsilon_{k, l}\boldsymbol{G}_{\mathcal{M}(l)}}\right)^{-1}
            \right], 
        \end{equation}
        which is independent of $i$ and only relies on $m$. 
        Thus, we define ${e}_{k, m} = \epsilon_{k, i}$ for $\mathcal{M}(i) = m$, $1\le m\le M$. 
        Given that $K\to\infty$, it is easy to verify that the cadinality of set $\mathcal{I}_{k, m} = \{i | i\neq k, 1\le i\le K, \mathcal{M}(i) = m\}$ satisfies 
        \begin{equation}
            \lim_{K\to\infty}{\frac{|\mathcal{I}_{k, m}|}{K - 1}} = \lim_{K\to\infty}{\frac{|\mathcal{I}_{k, m}|}{K}} = \mu_{m}, ~\forall k, m. 
        \end{equation}
        Then, the matrix $\boldsymbol{Y}_{k} - \boldsymbol{I}_{N}$ can be equivalently rewritten as 
        \begin{subequations}
            \begin{align}
                \boldsymbol{Y}_{k} - \boldsymbol{I}_{N} & = \sum_{i\neq k}{\epsilon_{k, i}\boldsymbol{G}_{\mathcal{M}(i)}} = \sum_{m = 1}^{M}\left(
                    \sum_{i\in\mathcal{I}_{k, m}}{\epsilon_{k, i}\boldsymbol{G}_{m}}
                \right) \\
                & = \sum_{m = 1}^{M}{{e}_{k, m}\boldsymbol{G}_{m}\cdot |\mathcal{I}_{k, m}|} ~ \sim ~ (K - 1)\boldsymbol{A}_{k}, 
            \end{align}
        \end{subequations}
        where $\boldsymbol{A}_{k} = \sum_{m = 1}^{M}{{e}_{k, m}\mu_{m}\boldsymbol{G}_{m}}\in\mathbb{C}^{N\times N}$. 
        Consequently, the vector $\boldsymbol{e}_{k} = [e_{k, 1}, \ldots, e_{k, M}]^{T}\in\mathbb{R}_{+}^{M\times 1}$ asymptotically satisfies the following equation for large $K$, 
        \begin{equation}
            \mathrm{tr}\left[
                e_{k, m}\boldsymbol{G}_{m}\left(
                    \boldsymbol{I}_{N} + (K - 1)\boldsymbol{A}_{k}
                \right)^{-1}
            \right] = 1. 
        \end{equation}
        Notably, the above equation holds true for $\boldsymbol{e}_{k}$, $\forall k$, while its solution is unique due to uniqueness of $\epsilon_{k, i}$, $\forall k, i$, indicating that $\boldsymbol{e}_{k} = \bar{\boldsymbol{e}} = [\bar{e}_{1}, \ldots, \bar{e}_{M}]^{T}\in\mathbb{R}_{+}^{M\times 1}$ are identical given $K$. 
        As such, we have $\boldsymbol{A}_{k} = \bar{\boldsymbol{A}} \triangleq \sum_{m = 1}^{M}{\bar{e}_{m}\mu_{m}\boldsymbol{G}_{m}}$, $\boldsymbol{Y}_{k} = \bar{\boldsymbol{Y}} \triangleq \boldsymbol{I}_{N} + (K - 1)\bar{\boldsymbol{A}}$, $\forall k$, and $\bar{e}_{m}^{-1} = \mathrm{tr}(\boldsymbol{G}_{m}\bar{\boldsymbol{Y}}^{-1})$, $\forall m$, yielding
        \begin{equation}
            c_{k}^{\infty} = \left[
                \mathrm{tr}\left(\boldsymbol{G}_{\mathcal{M}(k)}\bar{\boldsymbol{Y}}^{-1}\right)
            \right] = \bar{e}_{\mathcal{M}(k)}^{-1}, ~\forall k. 
        \end{equation}
        Hence, we have $c_{k} - \bar{e}_{\mathcal{M}(k)}^{-1}\rightarrow 0$ under conditions of Proposition~\ref{prop:asymp-approx}, where $\bar{e}_{1}, \ldots, \bar{e}_{M}$ can be solved via the equation below, 
        \begin{equation}
            \mathrm{tr}\left[
                \bar{e}_{m}\boldsymbol{G}_{m}\bar{\boldsymbol{Y}}^{-1}
            \right] = 1, ~\forall m. 
        \end{equation}
        
        Therefore, the random vector $\boldsymbol{c}$ can be asymptotically approximated by constants $\bar{e}_{1}, \ldots, \bar{e}_{M}$ in the large system limit, where $\bar{e}_{m}^{-1}$ can be interpreted as the asymptotic decorrelated channel power of a user in the $m$-th subregion from other $(K - 1)$ users, $\forall m$, which is identical for all users within the same subregion. 
        This is a direct result of the Gaussian mixture channel model, where the channel covariance matrix for each subregion is distinct. 
        To further characterize the correlation between users randomly distributed over all subregions, we define $\rho_{K}$ as the decorrelated channel power gain across the cell, which is given by the averaged asymptotic decorrelated channel power gain over subregions given $K$, i.e., $\rho_{K} = \sum_{m = 1}^{M}\mu_{m}\bar{e}_{m}^{-1}$. 
        Thus, it can be verified that 
        \begin{subequations}
            \begin{align}
                & \rho_{K} - \mathbb{E}_{\boldsymbol{H}}\left[c_{k} | K\right] = \sum_{m = 1}^{M}
                \mu_{m}\Big(\bar{e}_{m}^{-1} - \\
                & ~~~~ ~~~~ ~~~~ ~~~~ \mathbb{E}_{\boldsymbol{H}}\left[c_{k} | K, \mathcal{M}(k) = m\right]\Big)
                \rightarrow 0. 
            \end{align}
        \end{subequations}
        Moreover, by replacing $\bar{e}_{m}^{-1}$, $1\le m\le M$, with $\rho_{K}$, we have $\bar{\boldsymbol{A}}\approx\sum_{m = 1}^{M}\rho_{K}^{-1}\mu_{m}\boldsymbol{G}_{m} = \rho_{K}^{-1}\bar{\boldsymbol{G}}$ and thus 
        \begin{subequations}
            \begin{align}
                \rho_{K} & = \sum_{m = 1}^{M}\mu_{m}\mathrm{tr}\left(\boldsymbol{G}_{m}\bar{\boldsymbol{Y}}^{-1}\right) = \mathrm{tr}\left(\bar{\boldsymbol{G}}\bar{\boldsymbol{Y}}^{-1}\right) \\
                & \approx\mathrm{tr}\left(\bar{\boldsymbol{G}}\left(\boldsymbol{I}_{N} + \rho_{K}^{-1}(K - 1)\bar{\boldsymbol{G}}\right)^{-1}\right), 
            \end{align}
        \end{subequations}
        which yields equation~\eqref{def:corr-factor-fixed-point}. 

        Additionally, the existence and uniqueness of the positive solution $\rho_{K}$ to equation~\eqref{def:corr-factor-fixed-point} can be shown. 
        Using the eigenvalue decomposition, equation~\eqref{def:corr-factor-fixed-point} can be equivalently rewritten as $\xi_{K}(\rho_{K}) = 1$, where function $\xi_{K}(\rho)$ is defined in~\eqref{def:corr-facotr-xi-func}. 
        Note that $\xi_{K}(0) = N/(K - 1) > 1$, $\xi_{K}(\rho)\to 0 < 1$ as $\rho\to\infty$, and it is easy to verify that $\xi_{K}(\rho)$ monotonically decreases with $\rho\ge 0$. 
        Thus, the positive solution $\rho_{K}$ for equation $\xi_{K}(\rho) = 1$ exists and can be uniquely determined.

    \section{Convergence Proof for the Newton's Method}
    \label{appendix:newton-convergence}
        From equation~\eqref{def:corr-factor-xi-func-deriv}, it can be verified that $\partial{\xi_{K}(\rho)}/\partial{\rho} = J_{K}(\rho) < 0$ and $\partial^{2}{\xi_{K}(\rho)}/\partial^{2}{\rho} > 0$ for $\rho \ge 0$ given $K\ge 2$, i.e., $\xi_{K}(\rho)$ decreases with $\rho$ and is convex. 
        First, we show via mathematical induction that, for $0\le i\le I_{c}$, we always have
        \begin{equation}\label{eq:newton-itr-rho-range}
            0\le\rho_{K}^{(i)}\le\rho_{K}. 
        \end{equation}
        Obviously, this condition holds for $i = 0$, where $\rho_{K}^{(0)} = 0$ is employed as initialization. 
        Then, by assuming condition~\eqref{eq:newton-itr-rho-range} for some $i$, we have $\xi_{K}(\rho_{K}^{(i)})\ge\xi_{K}(\rho_{K}) = 1$ and thus
        \begin{equation}
            \rho_{K}^{(i + 1)} = \rho_{K}^{(i)} - J_{K}(\rho_{K}^{(i)})^{-1}\left(\xi_{K}(\rho_{K}^{(i)}) - 1\right)\ge\rho_{K}^{(i)}\ge 0, 
        \end{equation}
        because $J_{K}(\rho_{K}^{(i)}) < 0$. 
        Meanwhile, due to the convexity of $\xi_{K}(\rho)$, we have 
        \begin{subequations}
            \begin{align}
                \xi_{K}(\rho_{K}^{(i + 1)}) & \ge \xi_{K}(\rho_{K}^{(i)}) + J_{K}(\rho_{K}^{(i)})\left(\rho_{K}^{(i + 1)} - \rho_{K}^{(i)}\right) \\
                & = \xi_{K}(\rho_{K}^{(i)}) - \left(\xi_{K}(\rho_{K}^{(i)}) - 1\right) \\
                & = 1, 
            \end{align}
        \end{subequations}
        which equals $\xi_{K}(\rho_{K})$ and indicates $\rho_{K}^{(i + 1)}\le\rho_{K}$ as $\xi_{K}(\rho)$ is monotonically decreasing. 
        Thereby, we have shown that $0\le\rho_{K}^{(i + 1)}\le\rho_{K}$ holds in this case. 
        Hence, condition~\eqref{eq:newton-itr-rho-range} must hold throughout the iterations of the Newton's method. 

        Next, by leveraging $\xi_{K}(\rho_{K}) = 1$ and the iterative update equation~\eqref{def:newtons-iterations}, the error between $\rho_{K}^{(i + 1)}$ and $\rho_{K}$ is given by 
        \begin{subequations}\label{eq:newton-method-error-evolve}
            \begin{align}
                & \left|\rho_{K}^{(i + 1)} - \rho_{K}\right| = \left|
                    \rho_{K}^{(i)} - \rho_{K} - J_{K}(\rho_{K}^{(i)})^{-1}\left(\xi_{K}(\rho_{K}^{(i)}) - 1\right)
                \right| \label{subeq:newton-method-error-update} \\
                & ~~~~ = \left|\rho_{K}^{(i)} - \rho_{K}\right|\cdot\left|
                    1 - J_{K}(\rho_{K}^{(i)})^{-1}\cdot\frac{\xi_{K}(\rho_{K}^{(i)}) - \xi_{K}(\rho_{K})}{\rho_{K}^{(i)} - \rho_{K}}
                \right|. \label{subeq:newton-method-error-decay}
            \end{align}
        \end{subequations}
        Since $\xi_{K}(\rho)$ is continuous and smooth for $\rho\ge 0$, there exists $\bar{\rho}_{K}^{(i)}\in[\rho_{K}^{(i)}, \rho_{K}]$ such that 
        \begin{equation}\label{def:newton-error-decay-lagrange-deriv}
            J_K(\bar{\rho}_{K}^{(i)}) = \frac{\xi_{K}(\rho_{K}^{(i)}) - \xi_{K}(\rho_{K})}{\rho_{K}^{(i)} - \rho_{K}}. 
        \end{equation}
        Moreover, due to the convexity of $\xi_{K}(\rho)$, its derivative $J_{K}(\rho)$ increases with $\rho\ge 0$, leading to 
        \begin{equation}
            J_K(0)\le J_K(\rho_{K}^{(i)})\le J_K(\bar{\rho}_{K}^{(i)})\le J_K(\rho_{K}) < 0. 
        \end{equation}
        Thus, the second term in equation~\eqref{subeq:newton-method-error-decay} can be bounded as follows: 
        \begin{equation}\label{eq:newton-error-decay-bounds}
            0\le 1 - \frac{J_K(\bar{\rho}_{K}^{(i)})}{J_K(\rho_{K}^{(i)})} \le 1 - \frac{J_K(\rho_{K})}{J_K(0)} < 1. 
        \end{equation}
        By substituting~\eqref{def:newton-error-decay-lagrange-deriv} and~\eqref{eq:newton-error-decay-bounds} into~\eqref{eq:newton-method-error-evolve}, we have 
        \begin{subequations}
            \begin{align}
                \left|\rho_{K}^{(i + 1)} - \rho_{K}\right| & \le \left|\rho_{K}^{(i)} - \rho_{K}\right|\cdot\left|
                    1 - \frac{J_K(\rho_{K})}{J_K(0)}
                \right| \\
                & \le \rho_{K}\left|
                    1 - \frac{J_K(\rho_{K})}{J_K(0)}
                \right|^{i + 1}, 
            \end{align}
        \end{subequations}
        which vanishes for sufficiently large $i$. 
        This completes the proof of convergence of the Newton's method employed in Section~\ref{subsec:asymp-approx}.

    \section{Analysis for VMF-Type APS}
    \label{appendix:vmf-corr-analysis}
        \subsection{Derivation of Equation~\eqref{eq:vmf-corr-sinc} and~\eqref{def:vmf-corr-complex-dist}}
        \label{appendix-subsec:vmf-corr-derivation}
            With sufficiently dense anglar sampling over $\mathcal{S}_{+}$, i.e., $N_{E}$ and $N_{A}$ goes to infinity, we have 
            \begin{subequations}\label{eq:vmf-corr}
                \begin{align}
                    & \left[\bar{\boldsymbol{G}}\right]_{ni} = \sum_{l = 1}^{L_{0}}{
                        b_{l}\exp\left(
                            -j\bar{\boldsymbol{\kappa}}_{l}^{T}(\tilde{\boldsymbol{r}}_{n} - \tilde{\boldsymbol{r}}_{i})
                        \right)
                    } \label{subeq:vmf-corr-def} \\
                    & \xrightarrow{N_{E}, N_{A}\rightarrow +\infty} \frac{1}{B}\int_{
                        \bar{\boldsymbol{\kappa}}\in\mathcal{S}_{+}
                    }{
                        \exp\left(
                            -j\bar{\boldsymbol{\kappa}}^{T}(\boldsymbol{\delta}_{n,i} + j\boldsymbol{\nu})
                        \right)
                    }{
                        \mathrm{d}^{2}{S}
                    }, \label{subeq:vmf-corr-integral}
                \end{align}
            \end{subequations}
            where~\eqref{subeq:vmf-corr-integral} results from replacing the Riemann sum in~\eqref{subeq:vmf-corr-def} with integral. 
            Meanwhile, the normalization factor becomes $B\rightarrow\beta^{-1}\int_{\mathcal{S}_{+}}\exp(\boldsymbol{\nu}^{T}\bar{\boldsymbol{\kappa}}){\mathrm{d}^{2}S}$. 

            Given a large $\nu_{0}$ with $\nu_{z} > 0$, the vMF over the other hemisphere with $\kappa^{z} < 0$, i.e., $\mathcal{S}_{-} = \{\boldsymbol{\kappa} = [\kappa^{x}, \kappa^{y}, \kappa^{z}]^{T}|\|\boldsymbol{\kappa}\|_{2} = \kappa_{0}, \kappa^{z} < 0\}$, becomes negligible compared to $\mathcal{S}_{+}$. 
            Thereby, $\left[\bar{\boldsymbol{G}}\right]_{ni}$ can be approximately computed by integrating over the full sphere $\mathcal{S} = \{\boldsymbol{\kappa} = [\kappa^{x}, \kappa^{y}, \kappa^{z}]^{T}|\|\boldsymbol{\kappa}\|_{2} = \kappa_{0}\}$ as follows, 
            \begin{equation}\label{eq:vmf-corr-fullsphere-integral}
                \left[\bar{\boldsymbol{G}}\right]_{ni} \approx \frac{1}{B}\int_{
                    \bar{\boldsymbol{\kappa}}\in\mathcal{S}
                }{
                    \exp\left(
                        -j\bar{\boldsymbol{\kappa}}_{l}^{T}(\boldsymbol{\delta}_{n,i} + j\boldsymbol{\nu})
                    \right)
                }{
                    \mathrm{d}^{2}{S}
                }. 
            \end{equation}
            To calculate the integral in~\eqref{eq:vmf-corr-fullsphere-integral}, we leverage the following integral equation for real vector $\boldsymbol{\delta}\in\mathbb{R}^{3\times 1}$, i.e., 
            \begin{equation}\label{def:real-vmf-sphere-integral}
                \int_{
                    \bar{\boldsymbol{\kappa}}\in\mathcal{S}
                }{
                    \exp\left(
                        -j\bar{\boldsymbol{\kappa}}_{l}^{T}\boldsymbol{\delta}
                    \right)
                }{
                    \mathrm{d}^{2}{S}
                } = 4\pi\kappa_{0}^{2}\mathrm{sinc}(\kappa_{0}\|\boldsymbol{\delta}\|_{2}), 
            \end{equation}
            which is verified in~\cite{ref:hmimo-spherical-integral-ref}. 
            By applying the analytical continuation technique~\cite{ref:analytical-continuation-source, ref:analytical-continuation-example}, i.e., replacing $\boldsymbol{\delta}$ with a complex vector $\boldsymbol{\delta}' = \boldsymbol{\delta}_{n, i} + j\boldsymbol{\nu}$, we have 
            \begin{equation}
                \left[\bar{\boldsymbol{G}}\right]_{n,i} \approx \frac{4\pi\kappa_{0}^{2}}{B}\mathrm{sinc}(\kappa_{0}d_{n, i}), 
            \end{equation}
            where $d_{n,i}$ is given by~\eqref{def:vmf-corr-complex-dist}. 
            Note that $\|\boldsymbol{\delta}\|_{2}$ in~\eqref{def:real-vmf-sphere-integral} is not simply replaced by $\|\boldsymbol{\delta}'\|_{2} = \sqrt{\boldsymbol{\delta}'^{H}\boldsymbol{\delta}'}$ because hermitian transpose is not holomorphic, which should be changed to $d_{n, i} = \sqrt{\boldsymbol{\delta}'^{T}\boldsymbol{\delta}'}$ instead. 

        \subsection{Proof of Equation~\eqref{eq:vmf-corr-sparse-limit}}
        \label{appendix-subsec:vmf-corr-sparse-limit-proof}
            From equation~\eqref{def:vmf-corr-complex-dist}, rewrite $d_{n, i}$ as follows, 
            \begin{equation}
                d_{n, i} = \sqrt{
                    (\|\boldsymbol{\delta}_{n,i}\|_{2}^{2} - \nu_{0}^2) + j2\boldsymbol{\delta}_{n,i}^{T}\boldsymbol{\nu}
                }. 
            \end{equation}
            Let $\hat{\boldsymbol{\delta}}_{n, i} = \boldsymbol{\delta}_{n, i}/\|\boldsymbol{\delta}_{n, i}\|_{2}$, $a = \|\boldsymbol{\delta}_{n,i}\|_{2}^{2} - \nu_{0}^2$, and $b = 2\boldsymbol{\delta}_{n,i}^{T}\boldsymbol{\nu}$. 
            Without loss of generality, let $d_{n, i} = d_{\text{re}} + jd_{\text{im}}$, where 
            \begin{equation}\label{def:vmf-type-complex-dist-solved}
                d_{\text{re}} = \sqrt{(r + a) / 2}, ~d_{\text{im}} = \sqrt{(r - a) / 2}, 
            \end{equation}
            with $r = \sqrt{a^{2} + b^{2}}$. 
            As $\|\boldsymbol{\delta}_{n, i}\|_{2}\rightarrow +\infty$ with fixed $\hat{\boldsymbol{\delta}}_{n, i}$, it can be easily verified that $a, b, r\rightarrow +\infty$ and $ab^{-2}\rightarrow (2\hat{\boldsymbol{\delta}}_{n, i}^{T}\boldsymbol{\nu})^{-2}$. 
            Therefore, we have $d_{\text{re}} = \sqrt{(r + a)/2} \rightarrow +\infty$ and 
            \begin{subequations}
                \begin{align}
                    d_{\text{im}} & = \sqrt{\frac{r - a}{2}} = \left(
                        \frac{1}{2}\frac{
                            b^{2}
                        }{\sqrt{a^{2} + b^{2}} + a}
                    \right)^{\frac{1}{2}} \\
                    & = \frac{1}{\sqrt{2}}\left[
                        \left(
                            \frac{a^{2}}{b^{4}} + \frac{1}{b^2}
                        \right)^{\frac{1}{2}} + \frac{a}{b^2}
                    \right]^{-\frac{1}{2}} \rightarrow \hat{\boldsymbol{\delta}}_{n, i}^{T}\boldsymbol{\nu}. 
                \end{align}
            \end{subequations}
            Note that $\mathrm{sinc}(\kappa_{0}d_{n, i})$ can be equivalently rewritten as 
            \begin{equation}
                \mathrm{sinc}(\kappa_{0}d_{n, i}) = \frac{e^{j\kappa_{0}d_{\text{re}}}e^{-\kappa_{0}d_{\text{im}}} - e^{-j\kappa_{0}d_{\text{re}}}e^{\kappa_{0}d_{\text{im}}}}{2j\kappa_{0}(d_{\text{re}} + jd_{\text{im}})}, 
            \end{equation}
            where the numerator is bounded while the denominator goes to infinity. 
            Hence, we have $4\pi\kappa_{0}^{2}B^{-1}\mathrm{sinc}(\kappa_{0}d_{n, i})\rightarrow 0$.

        \subsection{Proof of Equation~\eqref{eq:vmf-corr-concentrated-limit}}
        \label{appendix-subsec:vmf-corr-concentrated-limit-proof}
            When $\nu_{0}\rightarrow +\infty$ given $\boldsymbol{\delta}_{n, i}$, we have $a\rightarrow -\infty$, $b, r\rightarrow +\infty$, and $ab^{-2}\rightarrow -(2\boldsymbol{\delta}_{n, i}^{T}\hat{\boldsymbol{\nu}})^{-2}$. 
            Contrary to the former case, it can be shown that $d_{\text{im}} = \sqrt{(r - a) / 2}\rightarrow +\infty$ while $d_{\text{re}}\rightarrow \boldsymbol{\delta}_{n, i}^{T}\hat{\boldsymbol{\nu}}$ instead, where $\hat{\boldsymbol{\nu}} = \boldsymbol{\nu}/\nu_{0}$ as defined in Section~\ref{subsec:discussions}. 
            Meanwhile, the normalization factor becomes
            \begin{equation}\label{eq:vmf-norm-integral}
                B\sim\beta^{-1}\int_{\mathcal{S}}{
                    \exp(\boldsymbol{\nu}^{T}\bar{\boldsymbol{\kappa}})
                }{\mathrm{d}^{2}S} = \frac{2\pi\kappa_{0}(e^{\kappa_{0}\nu_{0}} - e^{-\kappa_{0}\nu_{0}})}{\beta\nu_{0}},
            \end{equation}
            where $x\sim y$ means $x/y\rightarrow 1$. 
            Note that the integration is applied over the whole sphere $\mathcal{S}$ because the integrand over $\mathcal{S}_{-}$ is asymptotically negligible as $\nu_{0}\to +\infty$. 
            Moreover, it can be verified that 
            \begin{equation}\label{eq:vmf-corr-dist-imag-scale-order-with-ccntn}
                \frac{d_{\text{im}}}{\nu_{0}} = \left(
                    \frac{\sqrt{a^{2} + b^{2}} - a}{2\nu_{0}^{2}}
                \right)^{\frac{1}{2}} = 1 + \mathcal{O}\left(
                    \left(\frac{1}{\nu_{0}}\right)^{2}
                \right), 
            \end{equation}
            where the second equation is obtained by applying the Taylor expansion w.r.t. $(1/\nu_{0})$. 
            Thus, we have $d_{\text{im}} - \nu_{0}\rightarrow 0$ and
            \begin{subequations}
                \begin{align}
                    & \frac{1}{B}\mathrm{sinc}(\kappa_{0}d_{n, i}) \sim \frac{
                        \beta\nu_{0}\left(
                            e^{j\kappa_{0}d_{\text{re}}}e^{-\kappa_{0}d_{\text{im}}} - e^{-j\kappa_{0}d_{\text{re}}}e^{\kappa_{0}d_{\text{im}}}
                        \right)
                    }{
                        j4\pi\kappa_{0}^{2}(d_{\text{re}} + jd_{\text{im}})(e^{\kappa_{0}\nu_{0}} - e^{-\kappa_{0}\nu_{0}})
                    } \\
                    & ~~~~ ~~~~ \sim \frac{\beta}{4\pi\kappa_{0}^{2}}e^{-j\kappa_{0}\boldsymbol{\delta}_{n, i}^{T}\hat{\boldsymbol{\nu}}}\cdot\frac{\nu_{0}}{d_{\text{im}}}e^{\kappa_{0}(d_{\text{im}} - \nu_{0})}  \\
                    & ~~~~ ~~~~ \rightarrow \frac{\beta}{4\pi\kappa_{0}^{2}}\exp(-j\kappa_{0}\boldsymbol{\delta}_{n, i}^{T}\hat{\boldsymbol{\nu}}), 
                \end{align}
            \end{subequations}
            as $\nu_{0}\rightarrow +\infty$, which yields equation~\eqref{eq:vmf-corr-concentrated-limit}.

    \section{Derivation of Equations~\eqref{def:corr-factor-grads-wrt-positions} and~\eqref{def:corr-factor-grads-Smat}}
    \label{appendix:corr-factor-grads}
        By leveraging equation $\xi_{N}(\rho_{N}) = 1$, we have
        \begin{equation}
            0 = \frac{\mathrm{d}{\xi_{N}(\rho_{N})}}{\mathrm{d}{\boldsymbol{t}}} = \frac{\partial{\xi_{N}(\rho_{N})}}{\partial{\rho_{N}}}\frac{\mathrm{d}{\rho_{N}}}{\mathrm{d}{\boldsymbol{t}}} + \frac{\partial{\xi_{N}(\rho_{N})}}{\partial{\boldsymbol{t}}}, 
        \end{equation}
        where $\boldsymbol{t}\in\{\boldsymbol{x}, \boldsymbol{y}\}$, which can be equivalently written as
        \begin{equation}
            \frac{\mathrm{d}{\rho_{N}}}{\mathrm{d}{\boldsymbol{t}}} = -\left(
                \frac{\partial{\xi_{N}(\rho_{N})}}{\partial{\rho_{N}}}
            \right)^{-1}\frac{\partial{\xi_{N}(\rho_{N})}}{\partial{\boldsymbol{t}}}. 
        \end{equation}
        To solve the gradient of $\rho_{N}$ w.r.t. $\boldsymbol{t}$, we have to first compute the partial derivatives of $\xi_{N}(\rho)$ w.r.t. $\rho$ and $\boldsymbol{t}$. 

        More generally, given function $\xi_{K}(\rho)$, $\forall K$, its derivative can be equivalently written as 
        \begin{subequations}\label{eq:partial-xi-to-rho}
            \begin{align}
                \frac{\partial{\xi_{K}(\rho)}}{\partial{\rho}} & = J_{K}(\rho) = -\sum_{n = 1}^{N}{
                    \frac{\lambda_{n}}{(\rho + (K - 1)\lambda_{n})^{2}}
                } \\
                & = -\mathrm{tr}\left[
                    \bar{\boldsymbol{G}}\left(
                        \rho\boldsymbol{I}_{N} + (K - 1)\bar{\boldsymbol{G}}
                    \right)^{-2}
                \right] \\
                & = -\rho^{-2}\mathrm{tr}\left(
                    \bar{\boldsymbol{G}}\boldsymbol{\Upsilon}_{K}^{-2}
                \right), 
            \end{align}
        \end{subequations}
        where $\boldsymbol{\Upsilon}_{K} = \boldsymbol{I}_{N} + \rho^{-1}(K - 1)\bar{\boldsymbol{G}}\in\mathbb{C}^{N\times N}$. 
        Meanwhile, we have $\xi_{K}(\rho) = \mathrm{tr}(\rho^{-1}\bar{\boldsymbol{G}}\boldsymbol{\Upsilon}_{K}^{-1})$ and thus
        \begin{equation}\label{def:partial-xi-to-position}
            \frac{\partial{\xi_{K}(\rho)}}{\partial{{t}_{n}}} = \mathrm{tr}\left(
                \rho^{-1}\frac{\partial{
                    \bar{\boldsymbol{G}}
                }}{\partial{
                    {t}_{n}
                }}\boldsymbol{\Upsilon}_{K}^{-1}
            \right) - \mathrm{tr}\left(
                \rho^{-1}\bar{\boldsymbol{G}}\frac{\partial{
                    \boldsymbol{\Upsilon}_{K}^{-1}
                }}{\partial{{t}_{n}}}
            \right), 
        \end{equation}
        where $t\in\{x, y\}$. 
        Next, denote the two terms on the right-hand side of equation~\eqref{def:partial-xi-to-position} as 
        \begin{subequations}
            \begin{align}
                \mathcal{R}_{t, n}^{(1)} & = \mathrm{tr}\left(
                    \rho^{-1}\frac{\partial{
                        \bar{\boldsymbol{G}}
                    }}{\partial{
                        {t}_{n}
                    }}\boldsymbol{\Upsilon}_{K}^{-1}
                \right), \\
                \mathcal{R}_{t, n}^{(2)} & = - \mathrm{tr}\left(
                    \rho^{-1}\bar{\boldsymbol{G}}\frac{\partial{
                        \boldsymbol{\Upsilon}_{K}^{-1}
                    }}{\partial{{t}_{n}}}
                \right), 
            \end{align}
        \end{subequations}
        respectively, $\forall n$, such that $\partial{\xi_{K}(\rho)}/\partial{{t}_{n}} = \mathcal{R}_{t, n}^{(1)} + \mathcal{R}_{t, n}^{(2)}$. 
        Moreover, define vectors $\boldsymbol{\mathcal{R}}_{t}^{(1)} = [\mathcal{R}_{t, 1}^{(1)}, \ldots, \mathcal{R}_{t, N}^{(1)}]^{T}\in\mathbb{R}^{N\times 1}$ and $\boldsymbol{\mathcal{R}}_{t}^{(2)} = [\mathcal{R}_{t, 1}^{(2)}, \ldots, \mathcal{R}_{t, N}^{(2)}]^{T}\in\mathbb{R}^{N\times 1}$. 
        Thus, we have
        \begin{equation}
            \frac{\partial{\xi_{K}(\rho)}}{\partial{\boldsymbol{t}}} = \left[
                \frac{\partial{\xi_{K}(\rho)}}{\partial{{t}_{1}}}, \ldots, \frac{\partial{\xi_{K}(\rho)}}{\partial{{t}_{N}}}
            \right]^{T} = \boldsymbol{\mathcal{R}}_{t}^{(1)} + \boldsymbol{\mathcal{R}}_{t}^{(2)}. 
        \end{equation}

        \begin{lemma}\label{lemma:Gmat-grad-trace}
            Given arbitrary Hermitian matrix $\boldsymbol{A}\in\mathbb{C}^{N\times N}$, define vector $\boldsymbol{\mathcal{T}}$ as 
            \begin{subequations}
                \begin{align}
                    \mathcal{T}_{n} & = \mathrm{tr}\left(
                        \frac{\partial{
                            \bar{\boldsymbol{G}}
                        }}{\partial{
                            {t}_{n}
                        }}\boldsymbol{A}
                    \right), ~\forall n, \\
                    \boldsymbol{\mathcal{T}} & = [\mathcal{T}_{1}, \ldots, \mathcal{T}_{N}]^{T}\in\mathbb{R}^{N\times 1}. 
                \end{align}
            \end{subequations}
            Then, $\boldsymbol{\mathcal{T}}$ can be equivalently written as
            \begin{equation}\label{def:appendix-grad-trace-lemma-expression}
                \boldsymbol{\mathcal{T}} = 2\mathrm{Re}\left[
                    \mathrm{diag}\left(
                        \boldsymbol{A}\boldsymbol{S}^{t}
                    \right)
                \right], 
            \end{equation}
            where $\boldsymbol{S}^{t}\in\mathbb{C}^{N\times N}$ is given by equation~\eqref{def:corr-factor-grads-Smat} and is rewritten here as follows
            \begin{equation}\label{def:appendix-Smat-rewrite}
                \boldsymbol{S}^{t} = \bar{\boldsymbol{Q}}^{H}\mathrm{Diag}(j\bar{\boldsymbol{\kappa}}^{t})\mathrm{Diag}(\boldsymbol{b})\bar{\boldsymbol{Q}}, ~t\in\{x, y\}. 
            \end{equation}
        \end{lemma}
        \begin{proof}[Proof\textup{:}\nopunct]
            Define matrix $\boldsymbol{\mathcal{E}}_{n}^{t}\triangleq\partial{\bar{\boldsymbol{G}}}/\partial{{t}_{n}}\in\mathbb{C}^{N\times N}$, $\forall n$. Note that the elements of $\bar{\boldsymbol{G}}$ can be written as 
            \begin{equation}
                \left[\bar{\boldsymbol{G}}\right]_{mm'} = \sum_{l = 1}^{L_{0}}{
                    b_{l}\exp(-j\bar{\boldsymbol{\kappa}}_{l}^{T}(\tilde{\boldsymbol{r}}_{m} - \tilde{\boldsymbol{r}}_{m'}))
                }
            \end{equation}
            for $1\le m, m'\le N$, which leads to $\left[\boldsymbol{\mathcal{E}}_{n}^{t}\right]_{mm'} = 0$ for $m, m'\neq n$. 
            Besides, we have $[\bar{\boldsymbol{G}}]_{nn} = \beta$ is constant and thus $\left[\boldsymbol{\mathcal{E}}_{n}^{t}\right]_{nn} = 0$. 
            For $m\neq m' = n$, it can be verified that 
            \begin{subequations}
                \begin{align}
                    \left[\boldsymbol{\mathcal{E}}_{n}^{t}\right]_{mn} & = \sum_{l = 1}^{L_{0}}{
                        j\bar{\kappa}_{l}^{t}\cdot b_{l}\exp(-j\bar{\boldsymbol{\kappa}}_{l}^{T}(\tilde{\boldsymbol{r}}_{m} - \tilde{\boldsymbol{r}}_{n}))
                    } \\
                    & = \bar{\boldsymbol{q}}_{m}^{H}\mathrm{Diag}(\boldsymbol{b})\mathrm{Diag}(j\bar{\boldsymbol{\kappa}}^{t})\bar{\boldsymbol{q}}_{n}, 
                \end{align}
            \end{subequations}
            where $\bar{\boldsymbol{q}}_{m}\in\mathbb{C}^{L_{0}\times 1}$ denotes the $m$-th column of $\bar{\boldsymbol{Q}}$ while $\bar{\boldsymbol{\kappa}}^{t} = [\bar{\kappa}_{1}^{t}, \ldots, \bar{\kappa}_{L_{0}}^{t}]^{T}\in\mathbb{R}^{L_{0}\times 1}$ as defined in Section~\ref{subsec:gradient-ascent}. 
            Since $\bar{\boldsymbol{G}}$ is an Hermitian matrix, we have $\left[\boldsymbol{\mathcal{E}}_{n}^{t}\right]_{nm} = \left[\boldsymbol{\mathcal{E}}_{n}^{t}\right]_{mn}^{*}$. 
            Therefore, elements of matrix $\boldsymbol{\mathcal{E}}_{n}^{t}$ are all zeros except for non-diagonal ones in the $n$-th row and $n$-th column. 
            By defining 
            \begin{subequations}
                \begin{align}
                    \boldsymbol{s}_{n}^{t} & = \bar{\boldsymbol{Q}}^{H}\mathrm{Diag}(\boldsymbol{b})\mathrm{Diag}(j\bar{\boldsymbol{\kappa}}^{t})\bar{\boldsymbol{q}}_{n}\in\mathbb{C}^{N\times 1}, \\
                    \boldsymbol{\Omega}_{n}^{t} & = [\boldsymbol{0}_{N\times(n - 1)}, \boldsymbol{s}_{n}^{t}, \boldsymbol{0}_{N\times(N - n)}]\in\mathbb{C}^{N\times N}, 
                \end{align}
            \end{subequations}
            it can be verified that $\boldsymbol{\mathcal{E}}_{n}^{t}$ can be equivalently expressed as 
            \begin{equation}
                \boldsymbol{\mathcal{E}}_{n}^{t} = \boldsymbol{\Omega}_{n}^{t} + (\boldsymbol{\Omega}_{n}^{t})^{H}. 
            \end{equation}
            Specially, for the $n$-th diagonal element, we have $[\boldsymbol{\Omega}_{n}^{t}]_{nn} + [(\boldsymbol{\Omega}_{n}^{t})^{H}]_{nn} = 0 = [\boldsymbol{\mathcal{E}}_{n}^{t}]_{nn}$. 
            Meanwhile, by denoting $\boldsymbol{\alpha}_{n}\in\mathbb{C}^{N\times N}$ as the $n$-th column of the Hermitian matrix $\boldsymbol{A}$, we have $\boldsymbol{A} = [\boldsymbol{\alpha}_{1}, \ldots, \boldsymbol{\alpha}_{N}] = [\boldsymbol{\alpha}_{1}, \ldots, \boldsymbol{\alpha}_{N}]^{H}$. 
            Then, $\mathcal{T}_{n}$ can be simplified as 
            \begin{subequations}
                \begin{align}
                    \mathcal{T}_{n} & = \mathrm{tr}\left(\boldsymbol{\mathcal{E}}_{n}^{t}\boldsymbol{A}\right) = \mathrm{tr}\left(\boldsymbol{\Omega}_{n}^{t}\boldsymbol{A}\right) + \mathrm{tr}\left((\boldsymbol{\Omega}_{n}^{t})^{H}\boldsymbol{A}\right) \\
                    & = \boldsymbol{\alpha}_{n}^{H}\boldsymbol{s}_{n}^{t} + (\boldsymbol{s}_{n}^{t})^{H}\boldsymbol{\alpha}_{n} = 2\mathrm{Re}\left(\boldsymbol{\alpha}_{n}^{H}\boldsymbol{s}_{n}^{t}\right). 
                \end{align}
            \end{subequations}
            Note that $\boldsymbol{\alpha}_{n}^{H}\boldsymbol{s}_{n}^{t}$ is the $n$-th diagonal element of matrix $\boldsymbol{A}\boldsymbol{S}^{t}$, where $\boldsymbol{S}^{t} = [\boldsymbol{s}_{1}^{t}, \ldots, \boldsymbol{s}_{N}^{t}]$ is given by equation~\eqref{def:appendix-Smat-rewrite}. 
            Hence, it is justified that vector $\boldsymbol{\mathcal{T}}$ can be expressed as equation~\eqref{def:appendix-grad-trace-lemma-expression}. 

        \end{proof}

        Following Lemma~\ref{lemma:Gmat-grad-trace}, it can be easily verified that 
        \begin{equation}
            \boldsymbol{\mathcal{R}}_{t}^{(1)} = 2\rho^{-1}\cdot\mathrm{Re}\left[
                \mathrm{diag}\left(
                    \boldsymbol{\Upsilon}_{K}^{-1}\boldsymbol{S}^{t}
                \right)
            \right]. 
        \end{equation}
        Meanwhile, by noting that 
        \begin{subequations}
            \begin{align}
                \mathcal{R}_{t, n}^{(2)} & = -\mathrm{tr}\left(
                    \rho^{-1}\bar{\boldsymbol{G}}\boldsymbol{\Upsilon}_{K}^{-1}\frac{\partial{\boldsymbol{\Upsilon}_{K}}}{\partial{{t}_{n}}}\boldsymbol{\Upsilon}_{K}^{-1}
                \right) \\
                & = -\mathrm{tr}\left(
                    \rho^{-1}\bar{\boldsymbol{G}}\boldsymbol{\Upsilon}_{K}^{-1}\rho^{-1}(K - 1)\frac{\partial\bar{\boldsymbol{G}}}{\partial{{t}_{n}}}\boldsymbol{\Upsilon}_{K}^{-1}
                \right), ~\forall n, 
            \end{align}
        \end{subequations}
        we can obatin the expression for $\boldsymbol{\mathcal{R}}_{t}^{(2)}$ as 
        \begin{equation}
            \boldsymbol{\mathcal{R}}_{t}^{(2)} = - 2(K - 1)\rho^{-2}\cdot\mathrm{Re}\left[
                \mathrm{diag}\left(
                    \boldsymbol{\Upsilon}_{K}^{-1}\bar{\boldsymbol{G}}\boldsymbol{\Upsilon}_{K}^{-1}\boldsymbol{S}^{t}
                \right)
            \right]. 
        \end{equation}
        Thus, the gradient $\partial{\xi_{K}(\rho)}/\partial{\boldsymbol{t}}$ is given by 
        \begin{subequations}\label{eq:partial-xi-to-position-finale}
            \begin{align}
                & \frac{\partial{\xi_{K}(\rho)}}{\partial{\boldsymbol{t}}} = 
                \begin{aligned}[t]
                    & 2\rho^{-1}\cdot\mathrm{Re}\big[
                        \mathrm{diag}\big(
                            \boldsymbol{\Upsilon}_{K}^{-1}\boldsymbol{S}^{t} \\
                        & ~~~~ - \rho^{-1}(K - 1)\boldsymbol{\Upsilon}_{K}^{-1}\bar{\boldsymbol{G}}\boldsymbol{\Upsilon}_{K}^{-1}\boldsymbol{S}^{t}
                        \big)
                    \big]
                \end{aligned}
                \label{subeq:partial-xi-to-position-added} \\
                & ~~~~ = 2\rho^{-1}\cdot\mathrm{Re}\left[
                    \mathrm{diag}\left(
                        \left(
                            \boldsymbol{I}_{N} - \boldsymbol{\Upsilon}_{K}^{-1}\left(
                                \boldsymbol{\Upsilon}_{K} - \boldsymbol{I}_{N}
                            \right)
                        \right)
                        \boldsymbol{\Upsilon}_{K}^{-1}\boldsymbol{S}^{t}
                    \right)
                \right] \label{subeq:partial-xi-to-position-combined} \\
                & ~~~~ = 2\rho^{-1}\cdot\mathrm{Re}\left[
                    \mathrm{diag}\left(
                        \boldsymbol{\Upsilon}_{K}^{-2}\boldsymbol{S}^{t}
                    \right)
                \right], 
            \end{align}
        \end{subequations}
        where equation~\eqref{subeq:partial-xi-to-position-combined} is obtained by replacing $\rho^{-1}(K - 1)\bar{\boldsymbol{G}}$ with $\boldsymbol{\Upsilon}_{K} - \boldsymbol{I}_{N}$. 
        
        Based on the results above while letting $K = N$ and $\rho = \rho_{N}$ in equations~\eqref{eq:partial-xi-to-rho} and~\eqref{eq:partial-xi-to-position-finale}, we have $\boldsymbol{\Upsilon}_{N} = \boldsymbol{Y}_{N} = \boldsymbol{I}_{N} + \rho_{N}^{-1}(N - 1)\bar{\boldsymbol{G}}$ and the gradient of $\rho_{N}$ w.r.t. $\boldsymbol{t}$ can be obtained as equation~\eqref{def:corr-factor-grads-wrt-positions}, i.e., 
        \begin{subequations}
            \begin{align}
                \frac{\mathrm{d}{\rho_{N}}}{\mathrm{d}{\boldsymbol{t}}} & = - \left[
                    \left(
                        \frac{
                            \partial{\xi_{N}(\rho)}
                        }{
                            \partial{\rho}
                        }
                    \right)^{-1}
                    \frac{\partial{\xi_{N}(\rho)}}{\partial{\boldsymbol{t}}}
                \right]\Bigg|_{\rho = \rho_{N}} \\
                & = \frac{
                    2\rho_{N}\mathrm{Re}\left[
                        \mathrm{diag}\left(
                            \boldsymbol{Y}_{N}^{-2}\boldsymbol{S}^{t}
                        \right)
                    \right]
                }{
                    \mathrm{tr}\left(
                        \boldsymbol{Y}_{N}^{-2}\bar{\boldsymbol{G}}
                    \right)
                }. 
            \end{align}
        \end{subequations}


\bibliographystyle{IEEEtran} 
\bibliography{IEEEabrv, reference}

\end{document}